\documentclass[12pt,reqno]{amsart}
\usepackage{mathrsfs}
\usepackage{amssymb,amsthm,amsmath}
\usepackage[numbers,sort&compress]{natbib}
\usepackage{amssymb,amsmath}
\usepackage{amsfonts}
\usepackage{mathrsfs}
\usepackage{latexsym}
\usepackage{amssymb}
\usepackage{amsthm}
\usepackage{color}
\usepackage{pdfsync}
\usepackage{indentfirst}
\usepackage{graphics}
\usepackage{subfigure}
\usepackage{epsfig}
\usepackage{float}
\date{today}

\usepackage{hyperref}
\hypersetup{hypertex=true,
	colorlinks=true,
	linkcolor=blue,
	anchorcolor=g,
	citecolor=red}

\usepackage{amsmath}
\usepackage{amsthm}
   
   \newtheorem{remark}{Remark}[section]

      		\newtheorem{proposition}{Proposition}[section]
   		    
            \newtheorem{corollary}{Corollary}[section]

\newcommand{\R}{\mathbb R}

\makeatletter

\newcommand{\Rmnum}[1]{\expandafter\@slowromancap\romannumeral #1@}
\makeatother

\newcommand{\beq}{\begin{equation}}
\newcommand{\eeq}{\end{equation}}
\newcommand{\ben}{\begin{eqnarray}}
\newcommand{\een}{\end{eqnarray}}
\newcommand{\beno}{\begin{eqnarray*}}
\newcommand{\eeno}{\end{eqnarray*}}

\numberwithin{equation}{section}

\begin{document}
	
\title[Stability of standing periodic waves in MTM]{Stability of standing periodic waves \\ in the massive Thirring model}

	\author{Shikun~Cui}
\address[Shikun~Cui]{School of Mathematical Sciences, Dalian University of Technology, Dalian, 116024,  China; Department of Mathematics and Statistics, McMaster University, Hamilton, Ontario L8S 4K1, Canada}
\email{cskmath@163.com}
\author{Dmitry~E.~Pelinovsky}
\address[Dmitry~E.~Pelinovsky]{Department of Mathematics and Statistics, McMaster University, Hamilton, Ontario L8S 4K1, Canada }
\email{dmpeli@math.mcmaster.ca}

\date{\today}
\maketitle

\begin{abstract}
We analyze the spectral stability of the standing periodic waves in the massive Thirring model in laboratory coordinates. Since solutions of the linearized MTM equation are related to the squared eigenfunctions of the linear Lax system, the spectral stability of the standing periodic waves can be studied by using their Lax spectrum. Standing periodic waves are classified based on eight eigenvalues which coincide with the endpoints of the spectral bands of the Lax spectrum. Combining analytical and numerical methods, we show that the standing periodic waves are spectrally stable if and only if the eight eigenvalues are located either on the imaginary axis or along the diagonals of the complex plane.
\end{abstract}

\vspace{0.25cm}

{\bf Keywords:} massive Thirring model, standing periodic waves, spectral stability, squared eigenfunctions, Lax spectrum.

\vspace{0.25cm}

{\it This article is written in memory of D. Kaup for his many contributions to studies of stability of nonlinear waves by using the squared eigenfunctions, including the pioneering work on the MTM system in \cite{KL_1}.}

\section{Introduction}

The Massive Thirring model (MTM) is a mathematical model in quantum field theory. This model was first proposed by Thirring in \cite{T1} and has been used as an integrable example of the nonlinear Dirac equation in the space of one dimension \cite{BC}. Integrability of the MTM was first established in \cite{Michailov} and then explored in \cite{KN,KM,O}. 

We use the following normalized form of the MTM in the laboratory coordinates:
\begin{equation}\label{ini_1}
	\left\{
	\begin{array}{lr}
		i(u_t+u_x)+v+|v|^2u=0, \\
		i(v_t-v_x)+u+|u|^2v=0,
	\end{array}
	\right.
\end{equation}
along with the initial condition $(u,v)|_{t=0}=(u_0, v_0)$. The complete integrability of the MTM is due to the existence of the following Lax pair of linear equations \cite{Michailov}:
\begin{equation}\label{lax_1}
	\left\{
	\begin{array}{lr}
		\psi_{x}=L(u,v,\lambda)\psi, \\
		\psi_{t}=M(u,v,\lambda)\psi,
	\end{array}
	\right.
\end{equation}
where
$$L=\frac{i}{4}(\lambda^2-\frac{1}{\lambda^2})\sigma_3
-\frac{i\lambda}{2}
\left(\begin{array}{cc}
	0       & \bar{v}  \\
	v & 0 \\
\end{array}\right)
+\frac{i}{2\lambda}\left(
\begin{array}{cc}
	0       &\bar{u}  \\
	u & 0 \\
\end{array}
\right)
+\frac{i}{4}(|u|^2-|v|^2)\sigma_3
$$
and
$$M=\frac{i}{4}(\lambda^2+\frac{1}{\lambda^2})\sigma_3
-\frac{i\lambda}{2}\left(
\begin{array}{cc}
	0       &\bar{v}  \\
	v & 0 \\
\end{array}
\right)-\frac{i}{2\lambda}
\left(\begin{array}{cc}
	0       & \bar{u}  \\
	u & 0 \\
\end{array}\right)
-\frac{i}{4}(|u|^2+|v|^2)\sigma_3.
$$
Here, $i$ is the imaginary unit with $i^2=-1$, the bar represents the complex conjugate, and $\sigma_3$ is the third Pauli matrix, $\sigma_3 = {\rm diag}(1,-1)$. 

The inverse scattering transform (IST) method for the linear system 
(\ref{lax_1}) was developed formally in \cite{V} and rigorously 
in \cite{Lee,PS_1,HLC1}. The IST method was used to solve the initial-value problem on the infinite line and to establish the long-time scattering properties near the soliton solutions. Solutions in the quarter-plane were also obtained with the unified transform method in \cite{X1}. The initial-value problem for the MTM system (\ref{ini_1}) has been alternatively studied by using the contraction mapping and energy methods in Sobolev space $H^s(\R)$ with $s > \frac{1}{2}$ (see review in \cite{P_1}). Solutions of the MTM system in the space $L^2(\R)$ of low regularity  were studied in \cite{C1,Huh,ST_1,Zhang}.

For physical applications in the quantum field theory and quantum optics, 
it is important to study stability of the standing and traveling waves in the time evolution of the MTM system (\ref{ini_1}). Every standing wave solution can be extended as 
the traveling wave solution due to the Lorentz symmetry
\begin{equation}
\label{MTM-Lorentz}
\left[ \begin{matrix}
u(x,t) \\
v(x,t)  
\end{matrix} \right]  \;\; \mapsto \;\; 
\left[ \begin{matrix}
\left( \frac{1 - c}{1 + c} \right)^{1/4} u\left( \frac{x + ct}{\sqrt{1-c^2}}, \frac{t + cx}{\sqrt{1-c^2}} \right) \\
\left( \frac{1 + c}{1 - c} \right)^{1/4} v\left( \frac{x + ct}{\sqrt{1-c^2}}, \frac{t + cx}{\sqrt{1-c^2}} \right)
\end{matrix}\right], 
\qquad c \in (-1,1),
\end{equation}
which exists because the MTM system (\ref{ini_1}) is relativistically invariant. In addition, every standing wave solution can be translated 
in space, time, and complex phase due to the translational and rotational symmetries 
\begin{equation}
\label{MTM-symm}
\left[ \begin{matrix}
u(x,t) \\
v(x,t)  
\end{matrix} \right]  \quad \mapsto \quad 
\left[ \begin{matrix}
u(x+x_0,t+t_0) e^{\mathrm{i} \theta_0} \\
v(x+x_0,t+t_0) e^{\mathrm{i} \theta_0}
\end{matrix}\right],
\qquad x_0, t_0, \theta_0 \in \mathbb{R}.
\end{equation}

The simplest standing wave solution of the MTM system is the Dirac soliton (also known as the gap soliton):
\begin{equation}\label{one_soliton_1}
\left\{
\begin{array}{lr}
u(x,t) = i \alpha~ {\rm sech}\left( \alpha x - i\frac{\gamma}{2}\right) e^{-i \beta t},\\
v(x,t)=-i \alpha~ {\rm sech}\left(\alpha x +i\frac{\gamma}{2}\right) e^{-i \beta t},
\end{array}
\right.
\end{equation}
where $\alpha := \sin \gamma$, $\beta := \cos \gamma$, and 
$\gamma\in(0,\pi)$ is a free parameter. Spectral stability 
of Dirac solitons was established with the completeness 
of squared eigenfunctions \cite{KL_1} and has been used 
to study spectral stability of solitary waves in non-integrable 
Dirac equations \cite{BPZ,BC1,KS02}. Orbital stability of Dirac solitons 
in Sobolev space $H^1(\R)$ was obtained in \cite{PS_1_1} by 
using the higher-order energy of the MTM system (\ref{ini_1}). Orbital stability of Dirac solitons in a weighted 
$L^2(\R)$ space was obtained in \cite{CPS_1} with the B\"{a}cklund--Darboux transformation. Transverse instability of Dirac solitons 
in the two-dimensional generalizations of the MTM system (\ref{ini_1}) 
was shown in \cite{PS_2}.

Multi-soliton solutions to the MTM system have also been constructed with different algebraic methods in \cite{BG,BG1,BG2,David,Feng1} both on the zero and constant nonzero backgrounds. More recently, the MTM system was studied for the existence of rogue waves given by the rational solutions which arise 
on the constant nonzero background due to its modulational instability 
\cite{Feng2,Dega,He,Ye}. It is important to combine the study of multi-soliton 
and multi-rogue-wave solutions on the nonzero background with the proper 
stability analysis of the nonzero background. This step was missing in most 
of the previous publications, where algebraic methods have been used. Based on several model examples involving the constant nonzero background \cite{BCD,BCG,DLS} and the standing periodic waves \cite{CP_PRE,CPW_1,PW_1} (see also reviews in \cite{DLS-book,SPP}), we know that the space--time localization of the rogue waves is related to the instability growth rate of the background. If the background is modulationally stable, numerical simulations do not show the occurrence of large-amplitude rogue waves \cite{MBH}. 

\vspace{0.25cm}

{\em The main motivation for our work is to give a complete spectral stability 
	analysis of the standing wave solutions to the MTM system (\ref{ini_1}) which include the constant nonzero solutions.}

\vspace{0.25cm}

The standing wave solutions of the MTM system (\ref{ini_1}) are written 
in the form:
\begin{equation}
\label{standing-wave}
u(x,t)=U(x)e^{-i\omega t}, \qquad v(x,t)=V(x)e^{-i\omega t},
\end{equation} 
where $\omega \in \mathbb{R}$ is the frequency parameter 
and $(U,V) \in \mathbb{C}^2$ is the wave profile. To include the class of 
Dirac solitons (\ref{one_soliton_1}), we will consider here 
the standing waves satisfying the reduction $V = \bar{U}$. 

Figure \ref{fig00} presents the existence diagram of the standing waves on the parameter plane $(b,\omega)$, where 
\begin{equation}
\label{Ham-spatial}
b := -\omega(|U|^2 + |V|^2) -|U|^2 |V|^2 - (\bar{U}V + \bar{V}U)
\end{equation} 
is the $x$-independent parameter that corresponds to the Hamiltonian of the spatial dynamical system for $(U,V)$. The existence results for the standing waves are summarized as follows:
\begin{itemize}
	\item Region $\rm \Rmnum{1}$ for $b \in (-\infty,0)$ contains exactly one family of standing waves with the mapping $x \mapsto \arg(U) = -\arg(V)$ being monotonically increasing. 
	
	\vspace{0.2cm}
	
	\item Region $\rm \Rmnum{2}$ bounded by $b = 0$, $\omega \in [-1,1]$ (black line), $b = (1-\omega)^2$, $\omega \in (-\infty,1]$ (red line), and $b = (1+\omega)^2$, $\omega \in (-\infty,-1]$ (blue line) contains exactly one family of standing waves with the mapping $x \mapsto \arg(U) = -\arg(V)$  being bounded and periodic.
	
		\vspace{0.2cm}
		 
	\item Region $\rm \Rmnum{3}$ for $b \in (0,(1+\omega)^2)$, $\omega \in (-\infty,-1)$ contains exactly two families of standing waves, both have the mapping $x \mapsto \arg(U) = -\arg(V)$ monotonically increasing. 
	
		\vspace{0.2cm}
		
	\item Region $\rm \Rmnum{4}$ contains no families of standing waves. 
\end{itemize}

For each family of the standing waves with the mapping $x \mapsto \arg(U) = -\arg(V)$ being monotonically increasing, there is a symmetrically reflected family with the mapping $x \mapsto \arg(U) = -\arg(V)$ being monotonically decreasing. Strictly speaking, such standing waves of the form (\ref{standing-wave}) are not periodic in $x$ even though the mapping $x \mapsto |U| = |V|$ is bounded and periodic. For notational convenience, we still refer to these solutions loosely as {\em the standing periodic waves}.

The constant solutions of the MTM system (\ref{ini_1}) occur on the boundary of the region $\rm \Rmnum{2}$, that is, for $b = 0$, $\omega \in [-1,1]$ (black line), $b = (1-\omega)^2$, $\omega \in (-\infty,1]$ (red line), and $b = (1+\omega)^2$, $\omega \in (-\infty,-1]$ (blue line). The constant solution is zero in the first case and nonzero in the other two cases.

\begin{figure}[h]
	\centering
	\includegraphics[width=3.5in,height=2.45in]{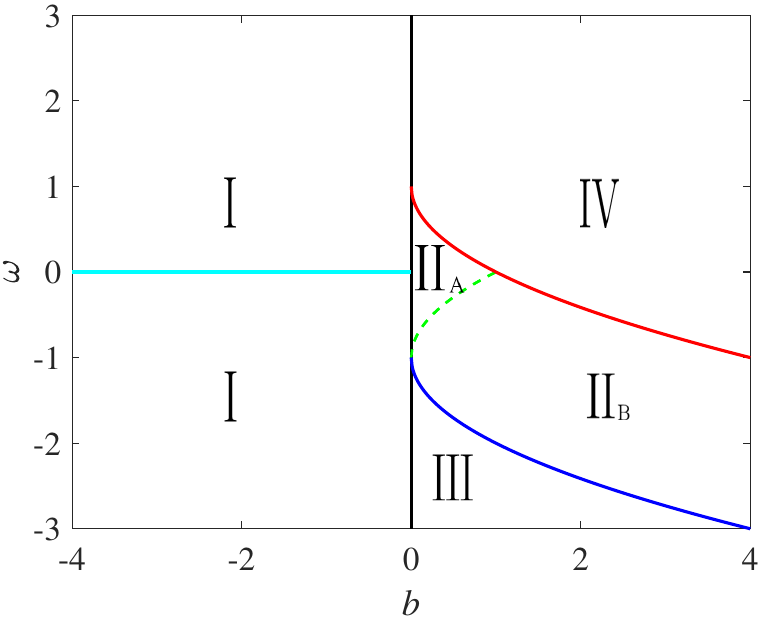}
	\caption{The existence diagram on the $(b,\omega)$ parameter plane.}
	\label{fig00}
\end{figure}

We define {\em the Lax spectrum} of the standing periodic waves as the set of admissible values of $\lambda$ in the linear system (\ref{lax_1}) for which $\psi(\cdot,t) \in L^{\infty}(\mathbb{R},\mathbb{C}^2)$ for every $t \in \R$. Accordingly, {\em the stability spectrum} is defined as the set of admissible values of $\Lambda$ in the linearized MTM system, see (\ref{line_11}) below, for which the eigenfunction is bounded on $\mathbb{R}$. By using the squared eigenfunction relation between solutions of the linear system (\ref{lax_1}) and solutions of the linearized MTM system (\ref{line_11}) found in \cite{KL_1}, we study the spectral stability of the standing periodic waves from their Lax spectrum. The spectral bands of the Lax spectrum which determine the spectral stability versus the spectral instability of the standing periodic waves are located between eight roots of the function
\begin{equation}
\label{polynomial-P}
P(\lambda) := \frac{1}{4} \left( \lambda^2 + \frac{1}{\lambda^2} - 2 \omega \right)^2 - b.
\end{equation}
Coefficients of $P(\lambda)$ are computed from parameters $(b,\omega)$ in  (\ref{standing-wave}) and (\ref{Ham-spatial}). We show that if $U = \bar{V}$, then the roots of $P(\lambda)$ satisfy the triple symmetry of reflections in the complex plane:
\begin{itemize}
	\item about the real axis $\R$, 
	\item about the imaginary axis $i \R$,
	\item about the unit circle $\mathbb{S}^1$.
\end{itemize}
By converting the linear system (\ref{lax_1}) to the matrix eigenvalue problem, see Appendix \ref{app_1}, we compute the Lax spectrum numerically by using the Fourier collocation method from \cite[Section 2.4]{Yjk}. 
The stability spectrum is obtained from the relation $\Lambda = \pm i \sqrt{P(\lambda)}$ due to the squared eigenfunction relation. 

Figure \ref{fig_total} displays the Lax spectrum (top panels) and the stability spectrum (bottom panels) for different families of the standing periodic waves. The location of the eight roots of $P(\lambda)$ is shown by red crosses. The dotted green line shows the unit circle $\mathbb{S}^1$. The location of roots of $P(\lambda)$, the Lax spectrum, and the stability spectrum are summarized as follows:

\begin{itemize}
	\item In region $\rm \Rmnum{1}$, the roots of $P(\lambda)$ form two quadruplets of complex eigenvalues which are symmetric about $\mathbb{S}^1$, see (a). The stability spectrum contains the unstable figure-eight band if $\omega \neq 0$, see (f). For $\omega = 0$, the bands connecting roots of $P(\lambda)$ are located along the main diagonals, see (b), and the stability spectrum is on $i \mathbb{R}$, see (g). 
	
			\vspace{0.2cm}
		
	\item In region $\rm \Rmnum{2}_A$, the roots of $P(\lambda)$ form two quadruplets of complex eigenvalues located on $\mathbb{S}^1$, see (c). 
	The stability spectrum contains the unstable segment on $\R$, see (h). 
	
			\vspace{0.2cm}
	
	\item In region $\rm \Rmnum{2}_B$, the roots of $P(\lambda)$ form a quadruplet of complex eigenvalues on $\mathbb{S}^1$ and two pairs of purely imaginary eigenvalues which are symmetric about $\mathbb{S}^1$, see (d). The stability spectrum contains the unstable segement on $\R$, see (i). 
	
				\vspace{0.2cm}
				
	\item In region $\rm \Rmnum{3}$, the roots of $P(\lambda)$ form four pairs of purely imaginary eigenvalues, which are symmetric about $\mathbb{S}^1$, see (e). The stability spectrum is on $i \R$, see (j).
\end{itemize}

\begin{figure}[h]
	\centering
	\subfigure[$\rm \Rmnum{1}$ with $\omega \neq 0$.
	]{\includegraphics[width=1.15in,height=1.15in]{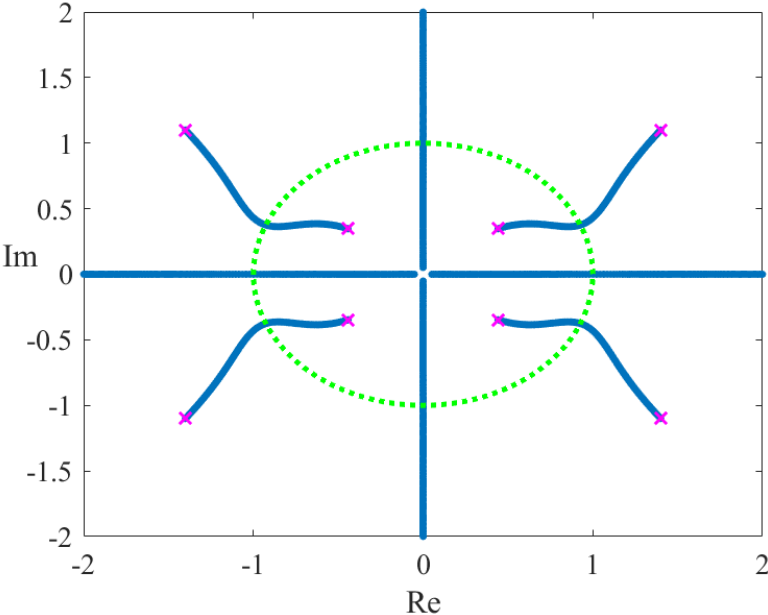}}
	\subfigure[$\rm \Rmnum{1}$ with $\omega = 0$. ]{\includegraphics[width=1.15in,height=1.15in]{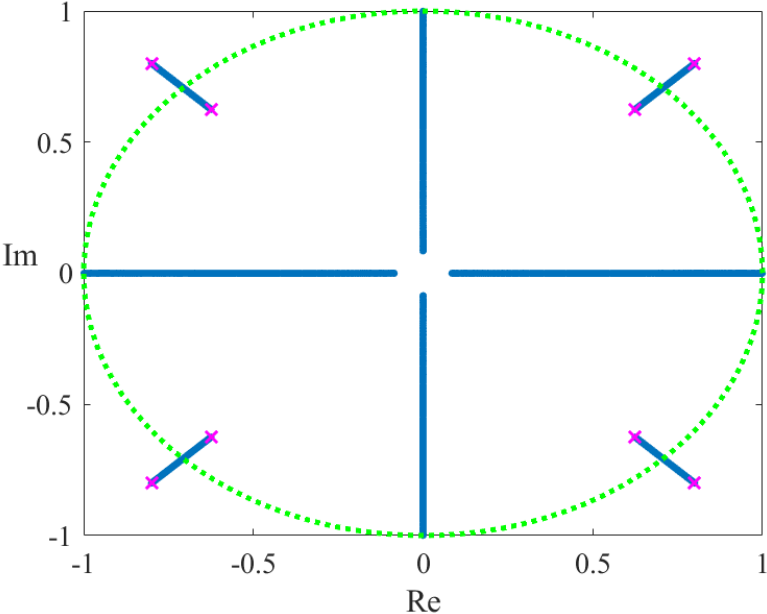}}
	\subfigure[$\rm \Rmnum{2}_A$. ]{\includegraphics[width=1.15in,height=1.15in]{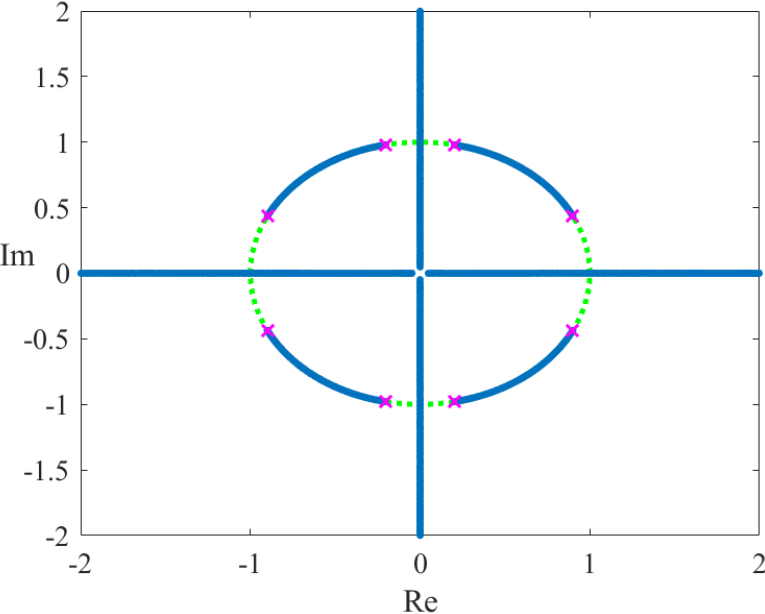}}
	\subfigure[$\rm \Rmnum{2}_B$. ]{\includegraphics[width=1.15in,height=1.15in]{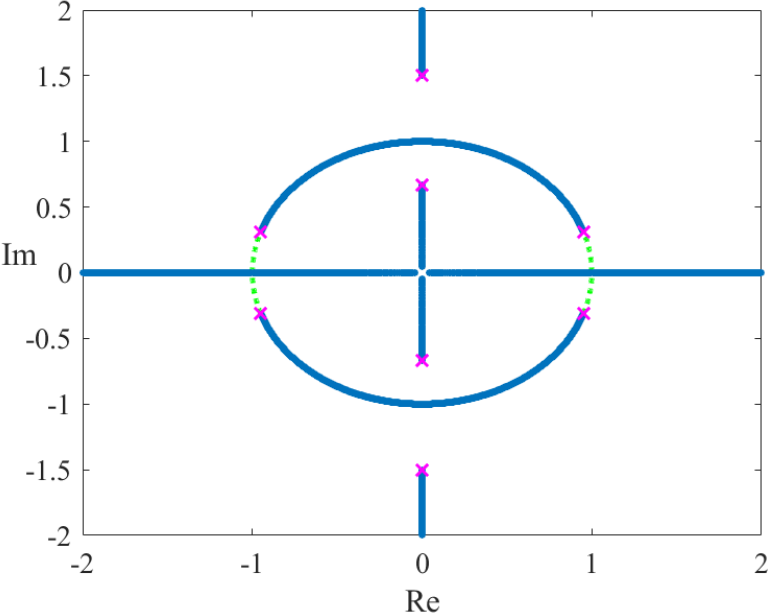}}
	\subfigure[$\rm \Rmnum{3}$. ]{\includegraphics[width=1.15in,height=1.15in]{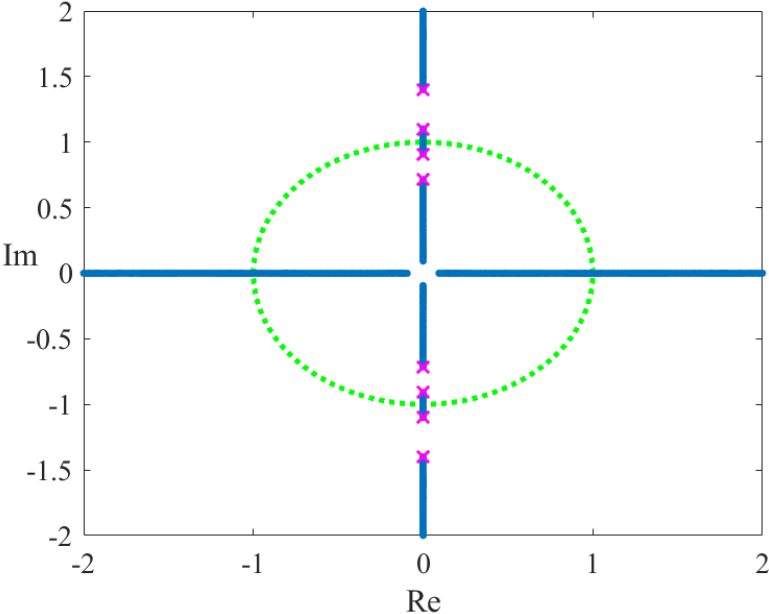}}
	\\
	\subfigure[$\rm \Rmnum{1}$ with $\omega \neq 0$.
	]{\includegraphics[width=1.15in,height=1.15in]{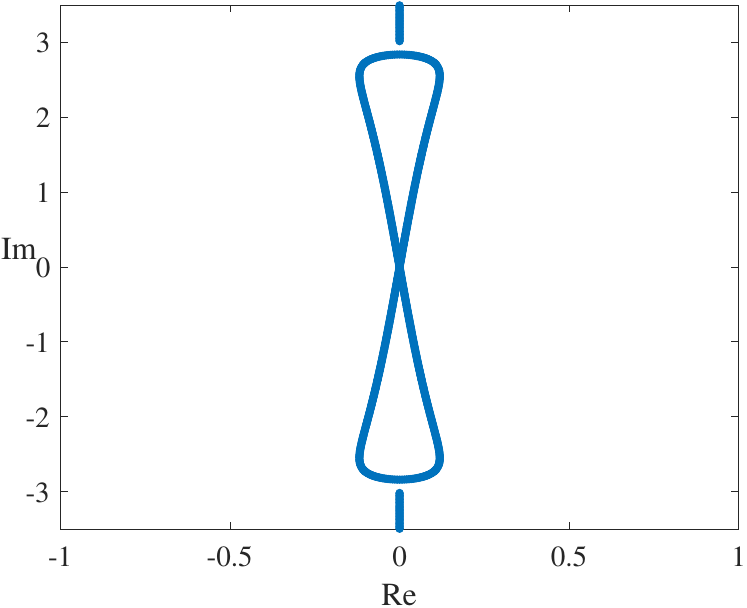}}
	\subfigure[$\rm \Rmnum{1}$ with $\omega = 0$.
	]{\includegraphics[width=1.15in,height=1.15in]{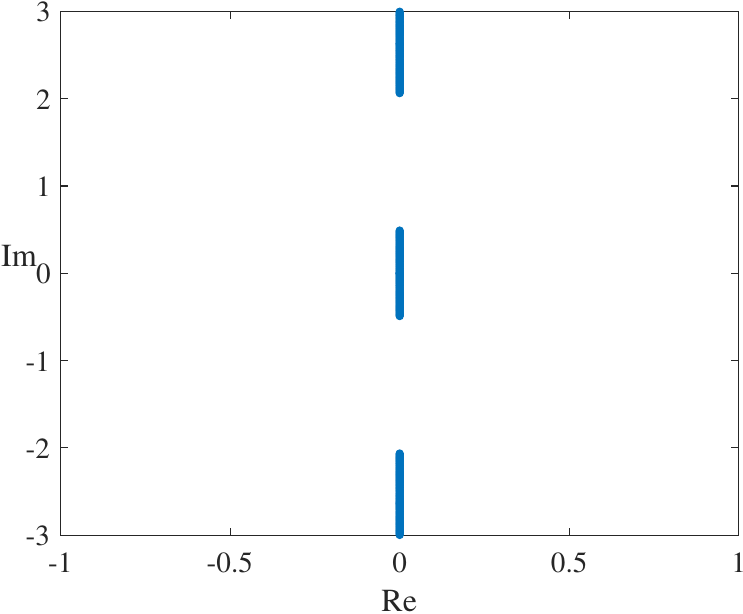}}
	\subfigure[$\rm \Rmnum{2}_A$.
	]{\includegraphics[width=1.15in,height=1.15in]{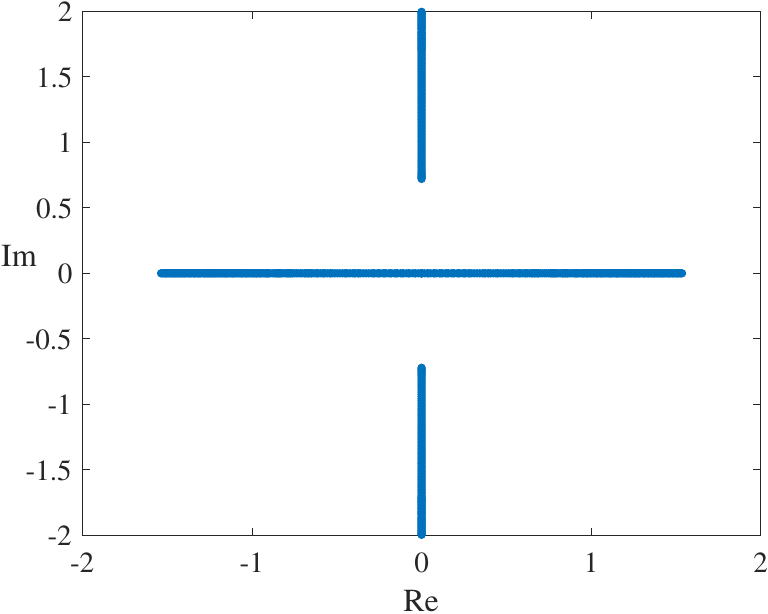}}
	\subfigure[$\rm \Rmnum{2}_B$.
	]{\includegraphics[width=1.15in,height=1.15in]{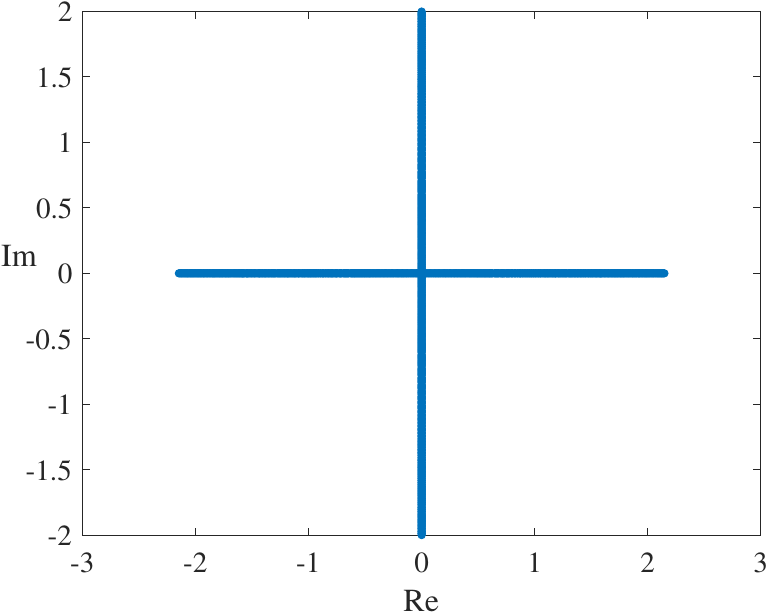}}
	\subfigure[$\rm \Rmnum{3}$.
	]{\includegraphics[width=1.15in,height=1.15in]{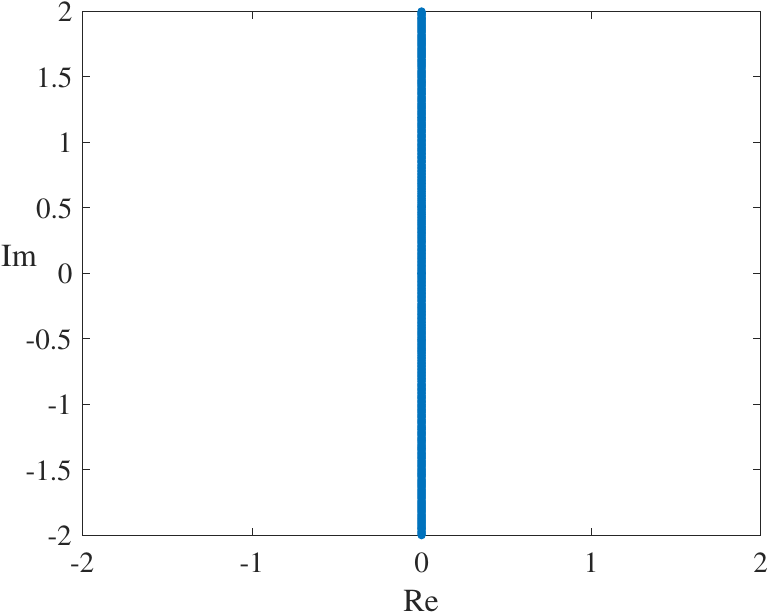}}
	\caption{Lax spectrum (top panels) and the stability spectrum (lower panels) for the standing periodic waves in regions $\rm \Rmnum{1}$, $\rm \Rmnum{2}_A$, $\rm \Rmnum{2}_B$, $\rm \Rmnum{3}$ of Figure \ref{fig00}.}
	\label{fig_total}
\end{figure}

We note that the Lax and stability spectra in region $\rm \Rmnum{3}$ are identical for both families of the standing periodic waves which coexist in region $\rm \Rmnum{3}$. Although Figure \ref{fig_total} only shows some numerical approximations for particular points in the $(b,\omega)$ plane, we have checked that the same results are true for every point in the corresponding regions of the parameter plane $(b,\omega)$. Based on these numerical approximations, we obtain the following stability criterion for the standing periodic waves of the MTM system (\ref{ini_1}), which is the main result of this work.

\vspace{0.25cm}

\centerline{\fbox{\parbox[cs]{\textwidth}{
The standing periodic waves in the form (\ref{standing-wave}) with $V = \bar{U}$ are spectrally stable in the MTM system (\ref{ini_1}) 
if and only if all roots of $P(\lambda)$ in (\ref{polynomial-P}) are located either 
on the imaginary axis $i \R$ or along the diagonals of the complex plane.
}}}

\vspace{0.25cm}

We also show the spectral stability of the constant nonzero background for $b = (1+\omega)^2$, $\omega \in (-\infty,-1]$ (blue line on Fig. \ref{fig00}) 
and the constant zero background for $b = 0$, $\omega \in [-1,1]$ (black line on Fig. \ref{fig00}). Moreover, the family of solitary waves 
on the constant nonzero background and the family of Dirac solitons (\ref{one_soliton_1}) on the constant zero background are also 
spectrally stable.  On the other hand, we show that the constant nonzero background 
for $b = (1-\omega)^2$, $\omega \in (-\infty,1]$ (red line on Fig. \ref{fig00}) is spectrally unstable. 

The study of spectral stability of standing and traveling wave solutions of integrable equations by using the squared eigenfunction method has started with the works of Deconinck and his coathors \cite{DMS,DS1,DU_1,SD,UD_1}. With the algebraic nonlinearization method of Cao and Geng \cite{cao1} which connects standing and traveling waves with the integrable finite-dimensional Hamiltonian systems, Chen and Pelinovsky found rogue wave solutions for many integrable equations 
\cite{CP-kdv,CP-nls,CP_1,CP_PRE} (see also \cite{CPW_2,PW_1}) in the cases when the wave background is modulationally unstable. The stability problem for the standing periodic waves can be solved from the Lax spectrum due to separation of variables and this has been explored 
for numerical study of the stability spectrum in many integrable equations \cite{CPU_1,CPW_1,CP_3}. However, the variables do not separate 
for the double-periodic solutions \cite{P_2} and for traveling periodic waves in lattice equations \cite{CP_2}. Further study of the spectral and orbital stability of the traveling wave solutions can be found in \cite{LS_1,LS_2}. Our work expands the study of spectral stability to the case of the standing periodic waves in the MTM system (\ref{ini_1}).

\vspace{0.25cm}

{\bf Organization of the paper.}
The standing waves of the form (\ref{standing-wave}) with $V = \bar{U}$ are classified in Section \ref{sec_class} by using the phase portraits for a planar Hamiltonian system. Section \ref{sec_modu} reports the squared eigenfunction relation between solutions of the linear system 
(\ref{lax_1}) and solutions of the linearized MTM system at the standing waves. 
Properties of eigenvalues of the Lax and stability spectra are described in Section \ref{sec_stability}.
The Lax and stability spectra for the constant nonzero solutions are 
computed explicitly in Section \ref{sec_laxpspec}. With these exact solutions, 
we have also tested the numerical method to recover the same 
spectra numerically. In Section \ref{sec_periodic_1}, we obtain 
exact solutions for the standing periodic waves in relation to roots of $P(\lambda)$ and compute numerical approximations of their Lax and stability spectra. The paper is concluded with a summary and further discussions in Section \ref{conclusion}. The numerical method is described in Appendix \ref{app_1}.

\vspace{0.25cm}

{\bf Acknowledgements.} The work of the first author was conducted during  PhD studies while visiting McMaster University. The first author thanks Professor Wendong Wang for the encouragement, the China Scholarship Council for financial support, and McMaster University for hospitality. The work of the second author was supported in part by the National Natural Science Foundation of China (No. 12371248).

\section{ Classification of standing waves}
\label{sec_class}

Profiles $(U,V)$ of the standing wave solutions of the form (\ref{standing-wave}) are found from the system of first-order differential equations 
\begin{equation}\label{ini_3}
\left\{
\begin{array}{lr}
iU' + \omega U + V +|V|^2U=0, \\
-iV' + \omega V + U + |U|^2V=0,
\end{array}
\right.
\end{equation}
which is obtained by substituting (\ref{standing-wave}) into (\ref{ini_1}).
System (\ref{ini_3}) can be written as the complex Hamiltonian system
\begin{equation}
\label{Hami_11}
i \frac{dU}{dx} = \frac{\partial H}{\partial \bar{U}}, \quad 
-i \frac{dV}{dx} = \frac{\partial H}{\partial \bar{V}},
\end{equation}
generated by the real-valued Hamiltonian
\begin{equation}
\label{H}
H(U,V)=-\omega(|U|^2+|V|^2)-|U|^2|V|^2-(\bar{U}V+\bar{V}U),
\end{equation}
which coincides with (\ref{Ham-spatial}). Since $H$ is independent of $x$, the Hamiltonian is conserved for every solution of system (\ref{ini_3}). Another real-valued conserved quantity for system (\ref{ini_3}) is 
\begin{equation}
\label{F}
F(U,V)=|U|^2-|V|^2,
\end{equation}
conservation of which follows by adding the following two equations
\begin{align*}
i(\bar{U} U' - \bar{U}' U) + \bar{U} V - U \bar{V} =& 0, \\
-i(\bar{V} V' - \bar{V}' V) + \bar{V} U - V \bar{U} =& 0.
\end{align*}
With two conserved quantities (\ref{H}) and (\ref{F}), system (\ref{ini_3}) is completely integrable. In what follows, we will only consider the standing waves under the reduction $V = \bar{U}$, which corresponds to $F(U,V) \equiv 0$. This particular case includes the Dirac solitons (\ref{one_soliton_1}) at the constant zero background. We use the polar form 
\begin{equation}
\label{periodic-waves}
U(x) = \zeta(x) e^{\frac{i}{2} \theta(x)}, \quad 
V(x) = \zeta(x) e^{-\frac{i}{2} \theta(x)}
\end{equation}
with real-valued $\zeta$ and $\theta$ and obtain the system of first-order differential equations
\begin{equation}
\label{hamiton_2}
\begin{cases}
\zeta' = \zeta \sin \theta,\\
\theta'= 2 \cos \theta + 2 \zeta^2 + 2\omega,
\end{cases}
\end{equation}
for which $b := H(U,V) = -2 \omega \zeta^2 - \zeta^4  - 2 \zeta^2 \cos \theta$ is a constant. With further transformation $\zeta = \sqrt{\xi}$, system (\ref{hamiton_2}) is rewritten in the form 
\begin{equation}\label{re_eq_2}
\begin{cases}
\xi' = 2 \xi \sin \theta,\\
\theta' = 2\omega+2\xi+2\cos\theta,
\end{cases}
\end{equation}
for which $b = -2\omega\xi-\xi^2-2\xi\cos\theta$. 

The relevant periodic solutions of system (\ref{re_eq_2}) correspond to the domain 
$$
\Gamma_+ := \left\{ (\theta,\xi) : \quad \xi \geq 0 \right\}
$$
and the line $\Gamma_0 := \{ (\theta,\xi) : \;\; \xi = 0 \}$ is invariant with respect to evolution of the spatial dynamical system (\ref{re_eq_2}). In addition, system (\ref{re_eq_2}) is $2\pi$-periodic with respect to $\theta$, which allows us to close the system on the cylinder $\mathbb{T}\times \mathbb{R}$, where $\mathbb{T}:= [0,2\pi)$ subject to the $2\pi$-periodicity condition. 

The following two propositions determine the equilibrium points of the planar system (\ref{re_eq_2}) in $\mathbb{T}\times \mathbb{R}$.

\begin{proposition}
	\label{prop-critical-points}
	System (\ref{re_eq_2}) admits the following equilibrium points in $\mathbb{T}\times \mathbb{R}$.
	
	\noindent$\bullet$ Two equilibrium points $\{\mathrm{P}_{+}, \mathrm{P}_{-}\}$ exist for every $\omega \in \R$, where 
	$$
	\mathrm{P}_{+} := \{ (\theta,\xi) = (0,-(1+\omega)) \} \quad \mbox{\rm and} 
	\quad \mathrm{P}_{-} := \{ (\theta,\xi) = (\pi,1-\omega) \}.
	$$
	
	\noindent$\bullet$ Two equilibrium points $\{ \mathrm{Q}_{+}, \mathrm{Q}_{-}\}$ exist for $\omega\in(-1,1)$, where 
	$$
	\mathrm{Q}_+ := \{ (\theta,\xi) = (\arccos(-\omega),0) \} \quad \mbox{\rm and} \quad \mathrm{Q}_- := \{ (\theta,\xi) = (2\pi-\arccos(-\omega),0) \}.
	$$	
\end{proposition}

\begin{proof}
	Assume that $(\theta_0, \xi_0) \in \mathbb{T}\times \mathbb{R}$ is the equilibrium point of system (\ref{re_eq_2}). Then 
\begin{equation*}
\begin{cases}
\xi_0 \sin \theta_0 = 0,\\
\omega_0 + \xi_0 + \cos\theta_0 = 0.
\end{cases}
\end{equation*}	
If $\xi_0=0$, then either $\theta_0={\rm arccos}(-\omega)$ or 
$\theta_0=2\pi-{\rm arccos}(-\omega)$. This yields $\{ \mathrm{Q}_{+}, \mathrm{Q}_{-}\}$ for every $\omega\in [-1,1]$.
If $\sin\theta_0=0$, then either $\theta_0 = 0$ and $\xi_0=-(1+\omega)$ 
or $\theta_0=\pi$ and $\xi_0=1-\omega$. This yields $\{ \mathrm{P}_{+}, \mathrm{P}_{-}\}$ for every $\omega \in \R$. For $\omega = \pm 1$, 
the sets $\{ \mathrm{Q}_{+}, \mathrm{Q}_{-}\}$ and $\{ \mathrm{P}_{+}, \mathrm{P}_{-}\}$ coincide.
\end{proof}

\begin{remark}
	The equilibrium points $\{ \mathrm{Q}_{+}, \mathrm{Q}_{-}\}$ belongs 
	to the invariant line $\Gamma_0$ for $\omega \in (-1,1)$. For applications of (\ref{re_eq_2}) in $\mathbb{T}\times \mathbb{R}_+$, the equilibrium point $\mathrm{P}_{+}$ is relevant for $\omega \in (-\infty,-1]$ and the equilibrium point $\mathrm{P}_-$ is relevant for $\omega \in (-\infty,1]$. No equilibrium points belong to $\mathbb{T}\times \mathbb{R}_+$ for $\omega \in (1,\infty)$.
\end{remark}

\begin{proposition}
	\label{prop-stability-points}
Classification of equilibrium points is given in the following table:
	\begin{table}[H]
		\centering
	\label{tab_0}
		\begin{tabular}{|c|c|c|c|}
			\hline
			Point   &$\omega \in (1,\infty)$ &$\omega\in(-1,1)$  &$\omega \in (-\infty,-1)$ \\ \hline
			$\mathrm{P}_{+}$ & center & center & saddle \\
			$\mathrm{P}_{-}$ & saddle & center & center    \\
			$\mathrm{Q}_{+}$ & - & saddle & - \\
			$\mathrm{Q}_{-}$ & - & saddle & - \\
			\hline
		\end{tabular}
		\caption{The type of equilibrium points of Proposition \ref{prop-critical-points}.}
	\end{table}
\end{proposition}

\begin{proof}
Let $(\theta_0,\xi_0) \in \mathbb{T} \times \mathbb{R}$ be an equilibrium point of system (\ref{re_eq_2}) and $(\theta_1,\xi_1) \in \mathbb{R} \times \mathbb{R}$ be a small perturbation. Linearized equations of system  (\ref{re_eq_2}) at $(\theta_0,\xi_0)$ are given by 
	\begin{equation}\label{re_eq_3}
	\begin{cases}
	\xi_1' = 2\sin\theta_0 \xi_1 + 2 \xi_0 \cos\theta_0 \theta_1,\\
	\theta' = 2 \xi_1 - 2 \sin\theta_0 \theta_1.
	\end{cases}
	\end{equation} 
The linearized system (\ref{re_eq_3}) is defined by the coefficient matrix
	\begin{equation*}
	A=2\left(\begin{array}{cc}
	\sin\theta_0       &\xi_0\cos\theta_0  \\
	1 & -\sin\theta_0\\
	\end{array}\right),
	\end{equation*}
for which ${\rm tr}(A) = 0$ and $\det(A) = -4 \left[ \sin^2 \theta_0 + \xi_0 \cos \theta_0 \right]$. The sign of  $\det(A)$ determines the type of the equilibrium point. It is a center if $\det(A) > 0$ and a saddle if $\det(A) < 0$. 

For $\mathrm{P}_{+}$, we have $\det(A) = 4(1+\omega)$ so that 
it is a center for $\omega > -1$ and a saddle for $\omega < -1$. 
For $\mathrm{P}_{-}$, we have $\det(A) = 4(1-\omega)$ so that 
it is a center for $\omega < 1$ and a saddle for $\omega > 1$. 
For $\mathrm{Q}_{\pm}$, we have $\det(A) = -4\sin^2\theta_0$ so that they are saddles for every $\omega \in (-1,1)$. 
\end{proof}

\begin{remark}
For $\omega = -1$, the equilibrium points $\mathrm{Q}_{\pm}$ coallesce with $\mathrm{P}_+$ and induce the change of the type of $\mathrm{P}_+$ 
from a center for $\omega > -1$ to a saddle for $\omega < -1$. For $\omega = 1$, the equilibrium points $\mathrm{Q}_{\pm}$ coallesce with $\mathrm{P}_-$ and induce the change of the type of $\mathrm{P}_-$ 
from a center for $\omega < 1$ to a saddle for $\omega > 1$. 
\end{remark}

Next we classify all admissible solutions of system (\ref{re_eq_2}) by constructing phase portraits of the dynamical system on the phase plane $(\theta,\xi)$ in $\mathbb{T} \times \mathbb{R}$. The classification of admissible solutions is summarized in the following proposition.

\begin{proposition}
	\label{pro_40}
The system (\ref{re_eq_2}) admits the following bounded solutions in $\mathbb{T} \times \mathbb{R}_+$:
	
$\bullet$ For $\omega \in (-\infty,-1)$, three families exist for
	$b\in (-\infty, (1+\omega)^2)$, $b\in( 0,(1+\omega)^2)$, and 
	$b\in ((1+\omega)^2, (1-\omega)^2)$, separated by two heteroclinic orbits 
	for $b = (1+\omega)^2$ from $\mathrm{P}_+$ and its $2\pi$-periodic continuation. The second family disappears at $\omega = -1$.
		
	$\bullet$ For $\omega\in(-1,1)$, two families exist for 
	$b\in (-\infty, 0)$ and $b\in (0, (1-\omega)^2)$ separated by two heteroclinic orbits for $b = 0$. The second family disappears at $\omega = 1$.

	$\bullet$ For $\omega \in (1,\infty)$, only one family exists for 
	$b\in (-\infty, 0).$
\end{proposition}

\begin{proof}
The assertion follows from the study of orbits of the planar Hamiltonian system on the phase plane $(\theta,\xi)$ in $\mathbb{T}\times \mathbb{R}$ shown on Figures \ref{fig_w_15}, \ref{fig_w_1}, and \ref{fig_w_16}. The phase portraits are obtained by plotting the level curves of the function $B(\theta,\xi) = b$, where 
$$
B(\theta,\xi) := -2 \omega \xi - \xi^2 - 2 \xi \cos \theta.
$$ 
Saddle points are shown by crosses and the center points are shown by stars.

For $\omega \in (-\infty,-1)$, see Figure \ref{fig_w_15}(a), 
we have $B(0,-(1+\omega)) = (1+\omega)^2$ for the saddle point $\mathrm{P}_+$ and $B(\pi,1-\omega)=(1-\omega)^2$ for the center point $\mathrm{P}_{-}$. $\mathrm{P}_-$ is the maximum of $B$ and $\mathrm{P}_+$ is the saddle point of $B$. One family of periodic orbits for $b\in ((1+\omega)^2, (1-\omega)^2)$ exists inside the punctured neightborhood of the center point $\mathrm{P}_-$ bounded by the two heteroclinic orbits from the saddle point $\mathrm{P}_+$ and its $2\pi$-periodic continuation. The second family of periodic orbits for $b \in (-\infty,(1+\omega)^2)$ exists above the upper heteroclinic orbit. The third family of periodic orbits for $b \in (0,(1+\omega)^2)$ exists between the lower heteroclinic orbit and the invariant line $\Gamma_0$. 

\begin{figure}[htb!]
	\centering
	\subfigure[$\omega<-1.$]{\includegraphics[width=2.9in,height=2.4in]{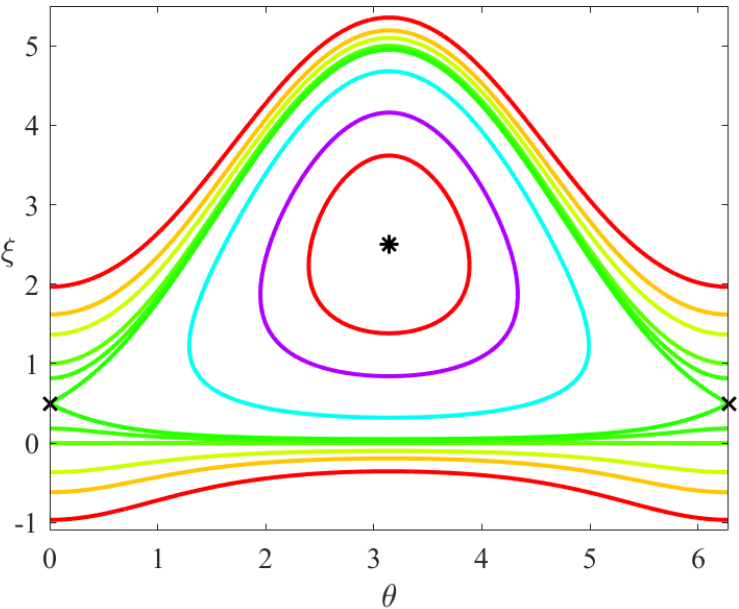}}
	\subfigure[$\omega>1.$]{\includegraphics[width=2.9in,height=2.4in]{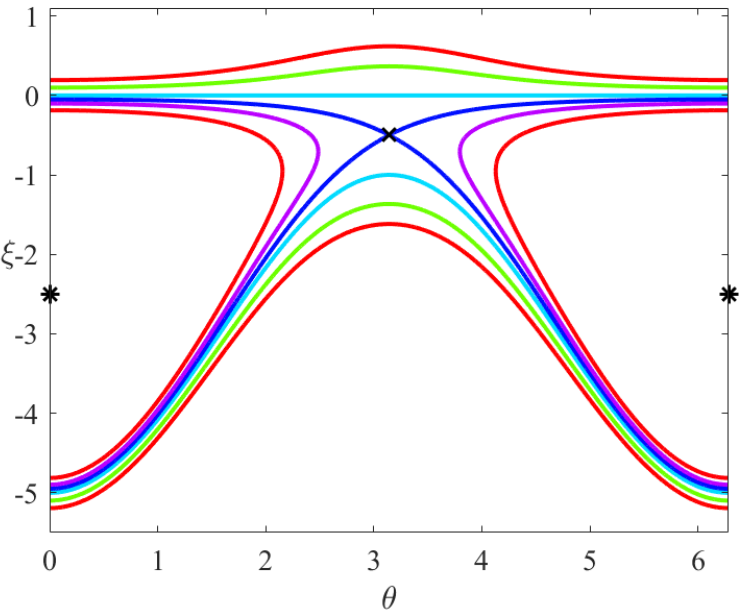}}
	\caption{Phase portraits in the phase plane $(\theta,\xi)$ for 
		(a) $\omega = -1.5$ and (b) $\omega = 1.5$. }\label{fig_w_15}
\end{figure}

When $\omega = -1$, see Figure \ref{fig_w_1}(a), the saddle point $\mathrm{P}_+$ and its $2\pi$-periodic continuation belongs to the invariant line $\Gamma_0$. The third family of periodic orbits between the lower heteroclinic orbit and the invariant line disappears whereas the other two families of periodic orbits remain.

\begin{figure}[htb!]
	\centering
	\subfigure[$\omega=-1.$]{\includegraphics[width=2.9in,height=2.4in]{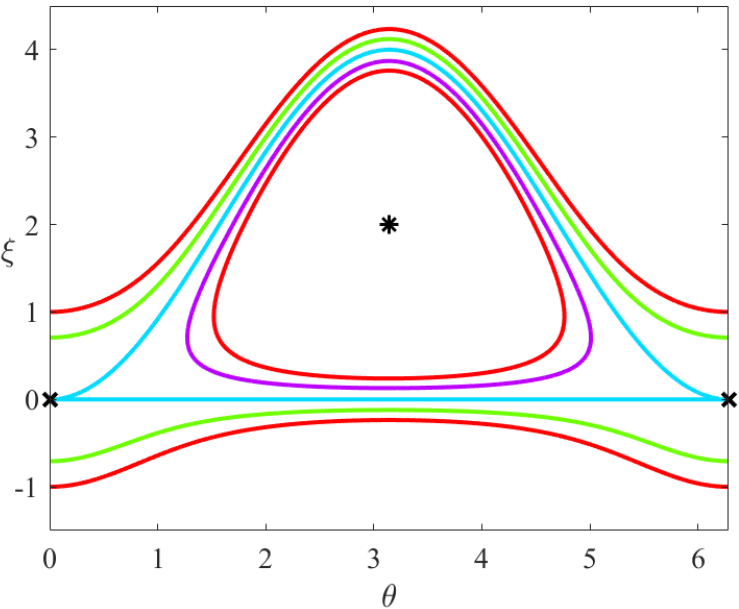}}
	\subfigure[$\omega=1.$]{\includegraphics[width=2.9in,height=2.4in]{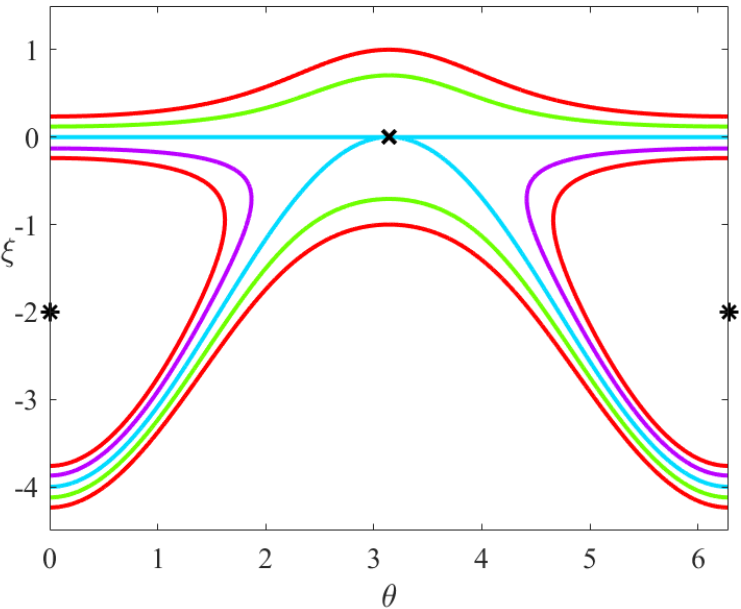}}
	\caption{Phase portraits in the phase plane $(\theta,\xi)$ for (a) $\omega=-1$ and (b) $\omega=1$. }\label{fig_w_1}
\end{figure}

For $\omega \in (1,\infty)$, see Figure \ref{fig_w_15}(b),  both $\mathrm{P}_{-}$ and  $\mathrm{P}_{+}$ are located in $\mathbb{T}\times \mathbb{R}_-$. One family of periodic orbits for $b\in (-\infty, 0)$ exists above the invariant line $\Gamma_0$. When $\omega = 1$, see Figure \ref{fig_w_1}(b), the saddle point $\mathrm{P}_-$ belongs to the invariant line but does not affect the existence of the family of periodic orbits for $b \in (-\infty,0)$. 

Finally, for $\omega \in (-1,1)$, see Figure \ref{fig_w_16}, both $\mathrm{P}_+$ and $\mathrm{P}_-$ are center points, but $\mathrm{P}_+ \in \mathbb{T} \times \mathbb{R}_-$ and $\mathrm{P}_- \in \mathbb{T} \times \mathbb{R}_+$. The other two equilibrium points $\mathrm{Q}_+$ and $\mathrm{Q}_-$ are saddle points located on the invariant line $\Gamma_0$. One family of periodic orbits for $b\in (0, (1-\omega)^2)$ exists inside the punctured neightborhood of the center point $\mathrm{P}_-$ bounded by the two heteroclinic orbits from the two saddle points $\mathrm{Q}_+$ and $\mathrm{Q}_-$. The second family of periodic orbits for $b \in (-\infty,0)$ exists above the upper heteroclinic orbit. 
\end{proof}

\begin{figure}[htb!]
	\centering
	\subfigure[$\omega\in(-1,0)$.]{\includegraphics[width=2.9in,height=2.4in]{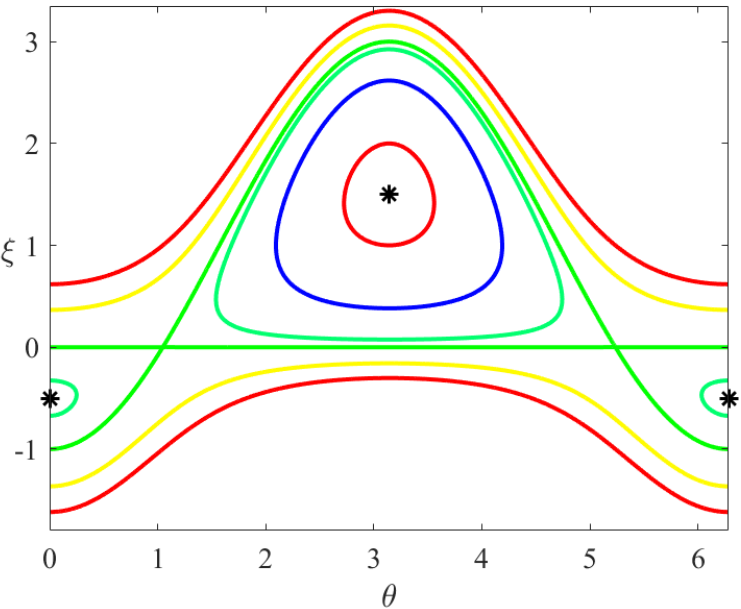}}
	\subfigure[$\omega\in(0,1)$.]{\includegraphics[width=2.9in,height=2.4in]{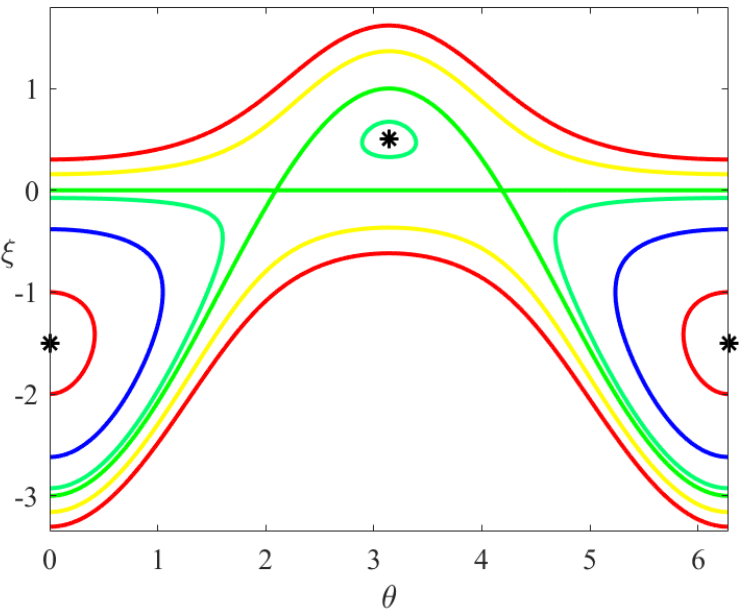}}
	\caption{Phase portraits in the phase plane $(\theta,\xi)$ for (a) $\omega=-0.5$ and (b) $\omega=0.5$. }\label{fig_w_16}
\end{figure}

\begin{remark}
	Existence of bounded solutions in Proposition \ref{pro_40} is represented in Figure \ref{fig00}, where the parameter plane $(b,\omega)$ is divided into several regions.  Region ${\rm \Rmnum{1}}$ for $b \in (-\infty,0)$ and $\omega \in \mathbb{R}$ contains one family of periodic orbits in $\mathbb{T} \times \mathbb{R}_+$ above the upper heteroclinic orbit. Region  ${\rm \Rmnum{2}}$ for either $b \in (0,(1-\omega)^2)$ and $\omega \in (-1,1)$ or $b \in ((1+\omega)^2,(1-\omega)^2)$ and $\omega \in (-\infty,-1]$ contains one family of periodic orbits in $\mathbb{T} \times \mathbb{R}_+$ inside the two heteroclinic orbits. Region  ${\rm \Rmnum{3}}$ for $b \in (0,(1+\omega)^2)$ and $\omega \in (-\infty,-1)$ contains two families of periodic orbits in $\mathbb{T} \times \mathbb{R}_+$ above the upper heteroclinic orbit and between the lower heteroclinic orbit and the invariant line $\Gamma_0$. Region  ${\rm \Rmnum{4}}$ contains no periodic orbits. 
\end{remark}

\begin{remark}
Boundaries between regions  ${\rm \Rmnum{1}}$,  ${\rm \Rmnum{2}}$,  ${\rm \Rmnum{3}}$, and  ${\rm \Rmnum{4}}$ in Figure \ref{fig00} correspond to some particular non-periodic solutions of system (\ref{re_eq_2}). The black line at $b = 0$ and $\omega \in [-1,1]$ gives the constant zero solution for the saddle points $\mathrm{Q}_+$ and $\mathrm{Q}_-$ and the solitary wave solutions on the zero background for the upper heteroclinic orbit (if $\omega \neq 1$). The blue line at $b = (1+\omega)^2$ and $\omega \in (-\infty,-1)$ gives the constant nonzero solution for the saddle point $\mathrm{P}_+$ and two solitary wave solutions on the nonzero background for the upper and lower heteroclinic orbits. The red line for $b = (1-\omega)^2$ and $\omega \in (-\infty,1)$ gives the constant nonzero solution for the center point $\mathrm{P}_-$.
\end{remark}

\section{Squared eigenfunction relation for the standing waves}
\label{sec_modu}

By substituting the standing waves of the form (\ref{standing-wave}) into the Lax system (\ref{lax_1}) and separating the variables with
\begin{equation}
\label{trans_1}
\psi(x,t) = e^{\frac{i}{2} \omega t \sigma_3 + \Omega t} \Psi(x),	
\end{equation}	
we obtain the following system of linear equations for $\Psi \in \mathbb{C}^2$ and $\Omega \in \mathbb{C}$,
\begin{equation}
\label{eig_1}
\left\{ \begin{array}{l} \Psi'(x) = L(U,V,\lambda) \Psi(x),\\
\Omega\Psi(x) = M(U,V,\lambda) \Psi(x) - 	\frac{i}{2} \omega \sigma_3\Psi(x).
\end{array} \right.
\end{equation}
The following proposition establishes the admissible values of $\Omega$ obtained from the characteristic function $P(\lambda)$ given by (\ref{polynomial-P}).

\begin{proposition}\label{pro_omega_1}
Let $(U,V)$ be solutions of (\ref{ini_3}) with $H(U,V) = b$ in (\ref{H}) and $F(U,V) = 0$ in (\ref{F}). Then, $\Omega$ is found from the 
characteristic equation 
\begin{equation}
4 \Omega^2 + P(\lambda) = 0,
\label{omega_1}
\end{equation}
where $	P(\lambda) $ is given by (\ref{polynomial-P}).
\end{proposition}

\begin{proof}
Since the second equation of system (\ref{eig_1}) is a linear homogeneous equation, there is a nonzero solution for $\Psi(x) \in \mathbb{C}^2$ if and only if 
$$
\left| \begin{array}{cc}
	-\frac{i}{2}(\lambda^2+\frac{1}{\lambda^2})+\frac{i}{2}(|U|^2+|V|^2)
	+i\omega  + 2 \Omega & \frac{i}{\lambda}\bar{U}+i\lambda\bar{V}  \\
	\frac{i}{\lambda}U+i\lambda V & \frac{i}{2}(\lambda^2+\frac{1}{\lambda^2})-\frac{i}{2}(|U|^2+|V|^2)-i\omega + 2 \Omega 
\end{array}
\right| = 0.
$$
The characteristic equation is written in the form (\ref{omega_1}) with 
\begin{align*}
P(\lambda) 
	&= \frac{1}{4}  \left( \lambda^2+\frac{1}{\lambda^2} - |U|^2 - |V|^2 - 2 \omega \right)^2
	+\left(\frac{1}{\lambda}\bar{U}+\lambda\bar{V} \right)\left(\frac{1}{\lambda}U+\lambda V\right) \\
	&= \frac{1}{4}  \left( \lambda^2+\frac{1}{\lambda^2} - 2 \omega \right)^2
	+ \frac{1}{2} (\lambda^2 - \frac{1}{\lambda^2}) (|V|^2 - |U|^2) \\
	& \quad 
	+ \frac{1}{4} (|U|^2 + |V|^2)^2 + \omega (|U|^2 + |V|^2) + \bar{V} U + V \bar{U}.
\end{align*}
If $H(U,V) = b$ and $F(U,V) = 0$, then $P(\lambda)$ is given by (\ref{polynomial-P}).
\end{proof}

To derive the spectral stability problem for the standing waves 
of the form (\ref{standing-wave}), we consider the time-dependent MTM system 
(\ref{ini_1}) and substitute the perturbation in the form
\begin{equation*}
		u(x,t) = e^{-i\omega t} [U(x) + \mathfrak{u}(x,t)], \quad 
		v(x,t) = e^{-i\omega t} [V(x) + \mathfrak{v}(x,t)].
\end{equation*}
Perturbation terms satisfy the linearized equations of motion,
\begin{equation}\label{line_10}
\left\{
\begin{array}{lr}
(i\partial_t + i \partial_x + \omega+|V|^2) \mathfrak{u} +(1+U\bar{V}) \mathfrak{v}  + UV \overline{\mathfrak{v}} = 0, \\
(-i\partial_t - i \partial_x + \omega+|V|^2) \overline{\mathfrak{u}} +(1+\bar{U} V ) \overline{\mathfrak{v}}  + \bar{U} \bar{V} \mathfrak{v}  = 0, \\
(i \partial_t -i\partial_x + \omega+|U|^2) \mathfrak{v} + (1+\bar{U}V) \mathfrak{u}  +UV \overline{\mathfrak{u}} = 0, \\
(-i \partial_t + i\partial_x + \omega+|U|^2)  \overline{\mathfrak{v}} + (1+U \overline{V})  \overline{\mathfrak{u}}  + \bar{U} \bar{V} \mathfrak{u} = 0.
\end{array}
\right.
\end{equation}
Normal modes are obtained after the separation of variables with 
\begin{align*}
\mathfrak{u}(x,t) = e^{\Lambda t} u_1(x), \quad 
\overline{\mathfrak{u}}(x,t) = e^{\Lambda t} u_2(x), \quad 
	\mathfrak{v}(x,t) = e^{\Lambda t} v_1(x),  \quad 
	\overline{\mathfrak{v}}(x,t) = e^{\Lambda t} v_2(x),
\end{align*}
where $u_2(x) \neq \bar{u}_1(x)$ and $v_2(x) \neq \bar{v}_1(x)$ if $\Lambda \notin \mathbb{R}$. The normal modes are found from the spectral stability problem, which follows from the linearized equations (\ref{line_10}),
\begin{equation}\label{line_11}
\left\{
\begin{array}{lr}
(~\ i\partial_x+\omega+|V|^2)u_{1}+(1+U\bar{V})v_1+UVv_2=-i\Lambda u_1, \\
(-i\partial_x+\omega+|V|^2)u_{2}+(1+\bar{U}V)v_2+\bar{U}\bar{V}v_1=i\Lambda u_2, \\
(-i\partial_x+\omega+|U|^2)v_{1}+(1+\bar{U}V)u_1+UVu_2=-i\Lambda v_1, \\
(~\ i\partial_x+\omega+|U|^2)v_{2}+(1+U\bar{V})u_2+\bar{U}\bar{V}u_1=i\Lambda v_2.
\end{array}
\right.
\end{equation} 
We define the admissible values of $ \Lambda $ from bounded solutions 
$(u_1, u_2, v_1, v_2) \in L^{\infty}(\mathbb{R},\mathbb{C}^4)$. If the admissible values of $\Lambda$ include points with ${\rm Re}(\Lambda)>0$, then the standing wave of the form (\ref{standing-wave}) is called {\em spectrally unstable}. Otherwise, it is called {\em spectrally stable}.

The central part of the spectral stability theory in the integrable systems 
is the relation between solutions of the linearized system (\ref{line_11}) 
and the squared eigenfunctions satisfying the linear system (\ref{eig_1}). 
This relation for the solitary wave solutions was found by D. Kaup and T. Lakoba in \cite[equation (18)]{KL_1}. The following proposition reproduces this result 
with a straightforward verification included for the sake of completeness.

\begin{proposition}
	\label{prop-squared}
Let $\Psi=(p,q)^{\mathrm{T}}$ be the eigenvector of the linear system
	(\ref{eig_1}) with some $\Omega \in \mathbb{C}$. Then 
	\begin{equation}
\left\{ \begin{array}{l} 
		u_1=\frac{1}{\lambda}q^2+Upq, \\
		u_2=\frac{1}{\lambda}p^2-\bar{U}pq,\\
		v_1=-\lambda q^2-Vpq,\\
		v_2=-\lambda p^2+\bar{V}pq
\end{array} \right.
	\end{equation}
is the solution of the linearized MTM system (\ref{line_11}) with 
\begin{equation}
\label{Lambda-Omega}
\Lambda=2\Omega=\pm {i}\sqrt{P(\lambda)}.
\end{equation}
\end{proposition}

\begin{proof}
Let $\Psi=(p,q)^{\mathrm{T}}$ be the eigenvector of the linear system
	(\ref{eig_1}) for some $\Omega \in \mathbb{C}$. By adding and subtracting the two equations, we obtain two linear systems 
	\begin{equation}\label{pq_1}
		\begin{aligned}
			i\left(\begin{array}{c}
				p'
				\\ q'
			\end{array}\right)
			-\left(\begin{array}{cc}
				\frac{\omega}{2}&0 \\
				0&-\frac{\omega}{2}
			\end{array}\right)	\left(\begin{array}{c}
				p
				\\ q
			\end{array}\right)+i\Omega\left(\begin{array}{c}
			p
			\\ q
		\end{array}\right)
			=  \left(\begin{array}{cc}
				\frac{|V|^2}{2}-\frac{\lambda^2}{2}&\lambda \bar{V} \\
				\lambda V &\frac{\lambda^2}{2}-\frac{|V|^2}{2}
			\end{array}\right)	\left(\begin{array}{c}
				p
				\\ q
			\end{array}\right),
		\end{aligned}
	\end{equation}
and 
\begin{equation}\label{pq_2}
	\begin{aligned}
		i\left(\begin{array}{c}
			p'
			\\ q'
		\end{array}\right)
		+\left(\begin{array}{cc}
			\frac{\omega}{2}&0 \\
			0&-\frac{\omega}{2}
		\end{array}\right)	\left(\begin{array}{c}
			p
			\\ q
		\end{array}\right)-i\Omega\left(\begin{array}{c}
		p
		\\ q
	\end{array}\right)
		=  \left(\begin{array}{cc}
			-\frac{|U|^2}{2}+\frac{1}{2\lambda^2}&-\frac{1}{\lambda} \bar{U} \\
			-\frac{1}{\lambda}{U} &\frac{|U|^2}{2}-\frac{1}{2\lambda^2}
		\end{array}\right)	\left(\begin{array}{c}
			p
			\\ q
		\end{array}\right).
	\end{aligned}
\end{equation}
Multiplying the first equations in systems (\ref{pq_1}) and (\ref{pq_2}) by $2p$ and the second equations in systems (\ref{pq_1}) and (\ref{pq_2}) by $2q$ yields
\begin{equation}\label{337}
	\left\{
	\begin{array}{l}
i(p^2)'-\omega p^2-|V|^2p^2+2i\Omega p^2=-\lambda^2p^2+2\lambda\bar{V}pq,\\
i(q^2)'+\omega q^2+|V|^2q^2+2i\Omega q^2=\lambda^2q^2+2\lambda Vpq, \\
i(p^2)'+\omega p^2+|U|^2p^2-2i\Omega p^2=\frac{1}{\lambda^2}p^2-\frac{2}{\lambda}\bar{U}pq,\\
i(q^2)'-\omega q^2-|U|^2q^2-2i\Omega q^2=-\frac{1}{\lambda^2}q^2-\frac{2}{\lambda}Upq.
	\end{array}
	\right.
\end{equation} 
Multiplying the first equations in systems (\ref{pq_1}) and (\ref{pq_2}) by $q$
and the second equations in systems (\ref{pq_1}) and (\ref{pq_2}) by $p$, and adding them together, yields
\begin{equation}\label{pqq_1}
	\left\{
\begin{array}{l}
i(pq)'=-2i\Omega pq+\lambda(\bar{V}q^2+ Vp^2), \\
i(pq)'=2i\Omega pq-\frac{1}{\lambda}({\bar{U}}q^2.
+{U}p^2).
	\end{array}
\right.
\end{equation}
By using (\ref{ini_3}) and (\ref{pqq_1}), we obtain 
\begin{equation}\label{33ZZ}
	\left\{
	\begin{array}{lr}
	i(Upq)'+\omega(Upq)+|V|^2(Upq)+2i\Omega Upq 
	=\lambda U\bar{V}q^2+\lambda UVp^2-Vpq,\\
	\vspace{0.2cm}
	i(\bar{U}pq)'-\omega(\bar{U}pq)-|V|^2(\bar{U}pq) +2i\Omega \bar{U}pq 
	=\lambda \bar{U}\bar{V}q^2+\lambda \bar{U}Vp^2+\bar{V}pq,\\
	\vspace{0.2cm}
	i(Vpq)'-\omega(Vpq)-|U|^2(Vpq)-2i\Omega Vpq 
	=-(\frac{V\bar{U}}{\lambda}q^2
	+\frac{VU}{\lambda}p^2)+Upq,\\
	\vspace{0.2cm}
	i(\bar{V}pq)'+\omega(\bar{V}pq)+|U|^2(\bar{V}pq)-2i\Omega \bar{V}pq
	=-(\frac{\bar{V}\bar{U}}{\lambda}q^2
	+\frac{\bar{V}U}{\lambda}p^2)-\bar{U}pq.
	\end{array}
	\right.
\end{equation} 
We are now ready to confirm all four equations of the linearized MTM system (\ref{line_11}) from four equations of (\ref{337}) and (\ref{33ZZ}):
\begin{equation*}
	\begin{aligned}
(i\partial_x+\omega+|V|^2) (\frac{1}{\lambda}q^2+Upq)&=-\frac{2i\Omega q^2}{\lambda}-2i\Omega Upq+\lambda q^2+ Vpq+\lambda U\bar{V}q^2+\lambda UVp^2\\
		&=-2i\Omega u_1-v_1-U\bar{V}v_1-UVv_2,
	\end{aligned}
\end{equation*}
\begin{equation*}
	\begin{aligned}
(i\partial_x-\omega-|V|^2) (\frac{1}{\lambda}p^2-\bar{U}pq)&=-\frac{2i\Omega p^2}{\lambda}+2i\Omega \bar{U}pq
		-\lambda p^2+ \bar{V}pq-\lambda \bar{U}\bar{V}q^2-\lambda \bar{U}Vp^2\\
		&=-2i\Omega u_2+v_2+\bar{U}\bar{V}v_1+\bar{U}Vv_2,
	\end{aligned}
\end{equation*}
\begin{equation*}
	\begin{aligned}
(i\partial_x-\omega-|U|^2) (\lambda q^2+Vpq)&=2i\lambda\Omega q^2+2i\Omega Vpq
		-\frac{1}{\lambda}q^2-Upq-(\frac{V\bar{U}}{\lambda}q^2
		+\frac{VU}{\lambda}p^2)\\
		&=-2i\Omega v_1-u_1-V\bar{U}u_1-VU u_2,		
	\end{aligned}
\end{equation*}
\begin{equation*}
	\begin{aligned}
(i\partial_x+\omega+|U|^2) (-\lambda p^2+\bar{V}pq)&=-2i\lambda\Omega p^2+2i\Omega \bar{V}pq
		-\frac{1}{\lambda}p^2+\bar{U}pq-(\frac{\bar{V}\bar{U}}{\lambda}q^2
		+\frac{\bar{V}U}{\lambda}p^2)\\
		&=2i\Omega v_2-u_2-\bar{V}\bar{U}u_1-\bar{V}Uu_2.		
	\end{aligned}
\end{equation*}
This yields (\ref{line_11}) with $\Lambda=2\Omega$. Then, $2\Omega = \pm i \sqrt{P(\lambda)}$ follows from (\ref{omega_1}).
\end{proof}

\section{Properties of eigenvalues of Lax and stability spectra}
\label{sec_stability}

Here we consider solutions of the spectral problem given by the first equation of system (\ref{eig_1}). If $\Psi = (p,q)^T$, then the spectral problem 
can be written explicitly as 
\begin{equation}
\label{KN}
\left\{ \begin{array}{l}
p'(x) = \frac{i}{4}(\lambda^2 - \frac{1}{\lambda^2}) p
-\frac{i\lambda}{2}\bar{V}(x) q
+\frac{i}{2\lambda}\bar{U}(x) q
+\frac{i}{4}(|U(x)|^2-|V(x)|^2)p, \\
q'(x) = -\frac{i}{4}(\lambda^2-\frac{1}{\lambda^2})q
-\frac{i\lambda}{2}V(x) p
+\frac{i}{2\lambda}U(x) p
-\frac{i}{4}(|U(x)|^2-|V(x)|^2)q. \end{array} \right.
\end{equation}
Recall that $(U,V)$ are bounded and that $V = \bar{U}$. The Lax spectrum of the spectral 
problem (\ref{KN}) is defined as the set of admissible values of $\lambda$ 
for which $\Psi = (p,q)^T$ is bounded. By Proposition \ref{prop-squared}, 
this solution of the spectral problem (\ref{KN}) defines a bounded solution $(u_1,u_2,v_1,v_2)$ of the spectral 
stability problem (\ref{line_11}) and hence the stability spectrum of the standing periodic waves. 

If $\zeta(x) = \zeta(x+L)$ is the $L$-periodic solution of system (\ref{hamiton_2}), then $\theta(x)$ is either $L$-periodic or monotonically increasing, see the phase portraits in the proof of Proposition \ref{pro_40}. 
The profiles $(U,V)$ in (\ref{periodic-waves}) are $L$-periodic in the former case and $L$-antiperiodic in the latter case. In both cases, 
the Floquet theorem can be used either with the period $L$ or with the period $2L$ to get all bounded solutions $(p,q)$ of the spectral problem (\ref{KN}) with $\lambda$ in the Lax spectrum. Each admissible value of $\lambda$ in the Lax spectrum will be referred to as {\em an eigenvalue} for simplicity. 

The following proposition specifies symmetries of eigenvalues in the Lax spectrum. 

\begin{proposition}
	\label{proposition_11}
Let $(U,V)$ be a solution of (\ref{ini_3}) and assume that $\lambda\in \mathbb{C} \backslash (\mathbb{R} \cup i \mathbb{R})$ is an eigenvalue of the spectral problem (\ref{KN}) with the eigenvector $\Psi=(p,q)^T$. Then
\begin{itemize}
	\item $-\lambda$ is also an eigenvalue with the eigenvector $\Psi=(p,-q)^T$. 
	
\item $\bar{\lambda}$ is also an eigenvalue with the eigenvector $\Psi = (\bar{q},-\bar{p})^T$. 
	
\item  $-\bar{\lambda}$ is also an eigenvalue with the eigenvector $\Psi = (\bar{q},\bar{p})^T$. 
\end{itemize}
\noindent If $V = \bar{U}$, then $\frac{1}{\lambda}$ is also an eigenvalue with the eigenvector $\Psi = (q,-p)^T$. 
\end{proposition}

\begin{proof}
Transformation $\lambda \to -\lambda$ and $(p,q) \to (p,-q)$ leaves system (\ref{KN}) invariant, which yields the first assertion.

Taking complex conjugate of system (\ref{KN}) yields
$$
\left\{ \begin{array}{l}
\bar{p}'(x) =-\frac{i}{4}(\bar{\lambda}^2
-\frac{1}{\bar{\lambda}^2})\bar{p}
+\frac{i\bar{\lambda}}{2} V \bar{q}
-\frac{i}{2\bar{\lambda}} U \bar{q}
-\frac{i}{4}(|U|^2-|V|^2)\bar{p}, \\
\bar{q}'(x) =\frac{i}{4}(\bar{\lambda}^2
-\frac{1}{\bar{\lambda}^2})\bar{q}
+\frac{i\bar{\lambda}}{2}\bar{V} \bar{p}
-\frac{i}{2\bar{\lambda}}\bar{U} \bar{p}
+\frac{i}{4}(|U|^2-|V|^2)\bar{q}. \end{array} \right.
$$	
Hence $(\bar{q},-\bar{p})^T$ is also a solution of   (\ref{KN}) with  $\lambda$ replaced by $\bar{\lambda}$, which yields the second assertion. 

The third assertion is a composition of the first two. 
	
If $V=\bar{U}$, then using (\ref{KN}) and replacing $\lambda$ with $\frac{1}{\lambda}$ yields
	$$
	\left\{ \begin{array}{l} q'(x) =\frac{i}{4}(\lambda^2-\frac{1}{\lambda^2})q
	-\frac{i}{2\lambda} V p +\frac{i\lambda}{2} U p
	=\frac{i}{4}(\lambda^2-\frac{1}{\lambda^2})q
	+\frac{i\lambda}{2}\bar{V} p
	-\frac{i}{2\lambda}\bar{U} p, \\
p'(x) =-\frac{i}{4}(\lambda^2-\frac{1}{\lambda^2})p
	-\frac{i}{2\lambda}\bar{V} q
	+\frac{i\lambda}{2}\bar{U} q
	=-\frac{i}{4}(\lambda^2-\frac{1}{\lambda^2})p
	+\frac{i\lambda}{2} V q
	-\frac{i}{2\lambda} U q. \end{array} \right.
	$$
Hence $(q,-p)^T$ is also a solution of (\ref{KN}) with  $\lambda$ replaced by $\frac{1}{\lambda}$, which yields the final assertion. 
\end{proof}

\begin{corollary}
	\label{cor-1}	
	If  $\lambda\in \mathbb{R}\backslash \{0 \} $ is an eigenvalue, then it is at least double with two eigenvectors $\Psi = (p,q)^T$ and $\Psi = (\bar{q},-\bar{p})^T$.
\end{corollary}

\begin{proof}
	If $\lambda\in\mathbb{R}\backslash \{0 \}$ is a simple 
	eigenvalue, then the symmetry in Proposition \ref{proposition_11} 
	implies that there is a constant $c_1\in\mathbb{C}$ such that 
	$$
	\left(\begin{array}{c}
	p         \\
	q  \\
	\end{array}\right)=c_1\left(\begin{array}{c}
	\bar{q}         \\
	-\bar{p}  \\
	\end{array}\right),
	$$ 
	which yields $p=c_1\bar{q}$, $\bar{q}=-\bar{c}_1p$, and hence $|c_1|^2=-1$, a contradiction. Therefore, $\lambda\in\mathbb{R}\backslash \{0 \}$ is at least a double eigenvalue.
\end{proof}

\begin{corollary}
	\label{cor-2}
If  $\lambda\in i\mathbb{R}\backslash \{0 \} $ is an eigenvalue, then 
it is simple if and only if the eigenvector $\Psi= (p,q)^T$ satisfies $p = c_2 \bar{q}$ for some constant $c_2 \in \mathbb{C}$ such that $|c_2|=1$.
\end{corollary}

\begin{proof}
	If $\lambda\in i\mathbb{R}\backslash \{0 \} $ is a simple eigenvalue, then the symmetry in Proposition \ref{proposition_11} implies that there is a constant $c_2 \in \mathbb{C}$ such that 
$$
\left(\begin{array}{c}
p         \\
q  \\
\end{array}\right) = c_2 
\left(\begin{array}{c}
\bar{q}         \\
\bar{p}  \\
\end{array}\right),
$$
which yields $p=c_2 \bar{q}$, $\bar{q}= \bar{c}_2 p$, and hence $|c_2|^2 = 1$. This gives the criterion for the eigenvalue $\lambda\in i\mathbb{R}\backslash \{0 \}$ to be simple. 
\end{proof}

Next we relate the Lax spectrum of the spectral problem (\ref{KN}) and the stability spectrum of the linearized MTM system (\ref{line_11}) 
by using the relation 
\begin{equation}
\label{stab-rel}
\Lambda = \pm i \sqrt{P(\lambda)}. 
\end{equation}
It follows from (\ref{polynomial-P}) that $P(\lambda)$ inherits the symmetry of Proposition \ref{proposition_11} since 
$V = \bar{U}$. If $\lambda \in \mathbb{C}$ is a root of $P(\lambda)$, so are $-\lambda$, $\lambda^{-1}$, and $-\lambda^{-1}$. Hence, let us introduce  $\lambda_1, \lambda_2 \in \mathbb{C}$ 
and factorize $P(\lambda)$ by its roots $\{ \pm \lambda_1, \pm \lambda_2, \pm \lambda_1^{-1}, \pm \lambda_2^{-2}\}$ 
\begin{equation}
\label{P-factor}
P(\lambda)=\frac{1}{4\lambda^4}(\lambda^2-\lambda_1^2)
(\lambda^2-\lambda_1^{-2})
(\lambda^2- \lambda_2^2)(\lambda^2 - \lambda_2^{-2}).
\end{equation}
The correspondence between (\ref{polynomial-P}) and (\ref{P-factor}) implies the relations between parameters $(b,\omega) \in \R^2$ and $(\lambda_1,\lambda_2) \in \mathbb{C}^2$:
\begin{equation}
\label{bb1}
\left\{ \begin{array}{r} 
4 \omega = \lambda_1^2 + \lambda_1^{-2} + \lambda_2^2 + \lambda_2^{-2}, \\
4 \omega^2 - 4 b = (\lambda_1^2 + \lambda_1^{-2}) (\lambda_2^2 + \lambda_2^{-2}).
\end{array} \right. 
\end{equation}

\begin{remark}
	If $\lambda_1 \in \mathbb{C} \backslash (\mathbb{R} \cup i \mathbb{R} \cup \mathbb{S}^1)$, then $\lambda_2 = \bar{\lambda}_1$ by the symmetry of Proposition \ref{proposition_11}. Also if $\lambda_1 \in \mathbb{R}$ such that $|\lambda_1| \neq 1$, then $\lambda_2 = \lambda_1$ by the symmetry of Corollary \ref{cor-1}. If either $\lambda_1 \in i \R$ or $\lambda_1 \in \mathbb{S}^1$, then $\lambda_2$ may be unrelated to $\lambda_1$. Further details on the distribution of roots of $P(\lambda)$ in relation to the periodic solutions of system (\ref{ini_3}) with $V = \bar{U}$ and parameters $(b,\omega) \in \R^2$ will be given in Section \ref{sec_periodic_1}.
\end{remark}

The following three propositions state some general results on the stability spectrum $\Lambda$ in (\ref{stab-rel}) from the roots of $P(\lambda)$ in (\ref{P-factor}) and the location of the Lax spectrum of $\lambda$. 

\begin{proposition}
	\label{prop-real}
	Assume that $\lambda_1,\lambda_2 \notin \R$ in (\ref{P-factor}). If $\lambda\in\mathbb{R}$, then $\Lambda\in i\mathbb{R}$. 
\end{proposition}

\begin{proof}
If $\lambda_1,\lambda_2 \notin \R$, then $P(\lambda)$ has no zeros for $\lambda \in \R$ and $P(\lambda) > 0$ for every $\lambda\in \mathbb{R}$ from the dominant term $P(\lambda) \sim \frac{1}{4} \lambda^4$ as $|\lambda| \to \infty$.
It follows from (\ref{stab-rel}) that $\Lambda \in i \R$ for $P(\lambda) > 0$. 
\end{proof}

\begin{proposition}
	\label{prop-imag}
	If $\lambda\in i\mathbb{R}$, then $\Lambda \in i\mathbb{R}$, provided the following conditions on the roots $\lambda_1$, $\lambda_2$ in (\ref{P-factor}) are satisfied:
	\begin{itemize}
		\item $\lambda_1,\lambda_2 \in \mathbb{C} \backslash (\mathbb{R} \cup i \mathbb{R})$.
		\item $\lambda_1 \in \mathbb{S}^1$ and $\lambda_2 = i \beta_2$ with $\beta_2 \geq 1$
under further restriction:
		\begin{equation}
		\label{restriction-1}
			{\rm Im(\lambda)\in(-\infty,-\beta_2]\cup[-\beta_2^{-1}, \beta_2^{-1}]\cup[\beta_2, \infty)}.
		\end{equation}
		\item $\lambda_1 = i \beta_1$ and $\lambda_2 = i \beta_2$ with $\beta_1 \geq \beta_2 \geq 1$ under further restriction:
		\begin{equation}
				\label{restriction-2}
			{\rm Im(\lambda)\in(-\infty,-\beta_1]
				\cup[-\beta_2, -\beta_2^{-1}]
				\cup[-\beta_1^{-1},\beta_1^{-1}]
				\cup[\beta_2^{-1}, \beta_2]\cup[\beta_1, \infty)}.
		\end{equation}
	\end{itemize}	
\end{proposition}

\begin{proof}
If $\lambda\in i\mathbb{R}$, we use $\lambda = i \beta$ and rewrite 
$P(\lambda)$ as 
	$$
	P(i \beta) = \frac{1}{4 \beta^4}(\beta^2 +\lambda_1^2)
	(\beta^2 + \lambda_1^{-2})
	(\beta^2 + \lambda_2^2)(\beta^2 + \lambda_2^{-2}).
	$$ 
If $\lambda_1,\lambda_2 \in \mathbb{C} \backslash (\mathbb{R} \cup i \mathbb{R})$, then $P(i \beta)$ has no roots for $\beta \in \R$ 
and $P(i \beta) > 0$ for every $\beta \in \R$ from the dominant term $P(i \beta) \sim \frac{1}{4} \beta^4$ as $|\beta| \to \infty$. It follows from (\ref{stab-rel}) that $\Lambda \in i \R$ for $P(i \beta) > 0$.
	
If $\lambda_1 \in \mathbb{S}$ and $\lambda_2 = i \beta_2$ with $\beta_2 \geq 1$, then 
	$$
	P(i \beta) = \frac{1}{4 \beta^4} (\beta^2+\lambda_1^2) (\beta^2+\lambda_1^{-2}) (\beta^2-\beta_2^2)(\beta^2-\beta_2^{-2}).
	$$ 
We have $P(i \beta) \geq 0$ if $\beta^2 \geq \beta_2^2$ or $0 \leq \beta^2 \leq \beta_2^{-2}$, which implies $\Lambda \in i \R$ under the restriction (\ref{restriction-1}). 
	
If $\lambda_1 = i \beta_1$ and $\lambda_2 = i \beta_2$ with $\beta_1 \geq \beta_2 \geq 1$, then 
$$
P(i \beta) = \frac{1}{4 \beta^4} (\beta^2-\beta_1^2) (\beta^2-\beta_1^{-2}) (\beta^2-\beta_2^2)(\beta^2-\beta_2^{-2}).
$$ 
We have $P(i \beta) \geq 0$ if $\beta^2 \geq \beta_1^2$, or $\beta_2^{-2} \leq \beta^2 \leq \beta_2^2$, or $0 \leq \beta^2 \leq \beta_1^{-2}$, which implies 
$\Lambda \in i \R$ under the restriction (\ref{restriction-2}).
\end{proof}

\begin{proposition}
	\label{prop-diag}
Assume that $\lambda_1 = \alpha_1 e^{\frac{\pi i}{4}}$, $\lambda_2 = \alpha_1 e^{-\frac{\pi i}{4}}$ with $\alpha_1 > 1$. If $\lambda = \pm \alpha e^{\pm \frac{\pi i}{4}}$ with $\alpha_1^{-1} \leq \alpha \leq \alpha_1$, then $\Lambda \in i \mathbb{R}$.
\end{proposition}

\begin{proof}
We rewrite (\ref{P-factor}) as
\begin{align*}
P\left(\pm \alpha e^{\pm \frac{\pi i}{4}} \right) = -\frac{1}{4\alpha^4} (\alpha^4 - \alpha_1^4) (\alpha^4 - \alpha_1^{-4}).
\end{align*}
We have $P(\lambda) \geq 0$ if $\alpha_1^{-4}\leq \alpha^4 \leq\alpha_1^4$, 
which is equivalent to $\alpha_1^{-1} \leq \alpha \leq \alpha_1$ if $\alpha > 0$. For this $\lambda$, we have $\Lambda \in i \R$. 
\end{proof}

\section{Lax and stability spectra for constant-amplitude solutions}
\label{sec_laxpspec}

Here we compute explicitly the Lax and stability spectra for the constant-amplitude solutions which correspond to the equilibrium points 
$\mathrm{P}_{\pm}$ in Proposition \ref{prop-critical-points}. The exact results are used for comparison with the numerical results to control accuracy of numerical approximations. We do not compute the Lax and stability spectra for 
the zero-amplitude solutions which correspond to the equilibrium points $\mathrm{Q}_{\pm}$ since the spectral stability of the zero-amplitude solutions is well-known for the MTM system (\ref{ini_1}) and the relevant Lax and stability spectra can be easily computed in the exact form.

\subsection{Constant-amplitude solution for $\mathrm{P}_-$}

By Proposition \ref{prop-stability-points}, the equilibrium point is a center for $\omega \in (-\infty,1)$, see Figures \ref{fig_w_15} (left), 
\ref{fig_w_1} (left), and \ref{fig_w_16}. It is located in $\mathbb{T} \times \mathbb{R}_+$  along the red curve on Figure \ref{fig00}, where $b = (1-\omega)^2$ and $\omega \in (-\infty,1]$. Since $\theta = \pi$ and $\xi = 1 - \omega$ by Proposition \ref{prop-critical-points}, the constant-amplitude solution  (\ref{periodic-waves}) is given by   
\begin{equation}\label{constant_1}
	U = i \sqrt{1-\omega},\quad
	V = - i\sqrt{1-\omega}, \quad \omega \in (-\infty,1).
\end{equation}
The Lax spectrum for the admissible solutions of the spectral problem (\ref{KN}) with $(U,V)$ in (\ref{constant_1}) is given by the following proposition.

\begin{proposition}\label{pro_71}
Let $(U,V)$ be given by (\ref{constant_1}) and consider bounded solutions 
of the spectral problem (\ref{KN}). 
\begin{itemize}
	\item  If $\omega \in (-\infty,0]$, then the Lax spectrum is given by 
\begin{align*}
\mathbb{R} \cup \left\{ i \beta : \; \beta \in (-\infty, -\beta_2 ]
	\cup [-\beta_2^{-1}, \beta_2^{-1}] \cup [\beta_2,+\infty) \right\} \cup\mathbb{S}^1 \backslash \{0\},	
\end{align*}
where $\beta_2 := \sqrt{1-\omega}+\sqrt{-\omega}$.

\item If $\omega \in (0,1)$, then the Lax spectrum is given by 
\begin{align*}
\mathbb{R}\cup i\mathbb{R}\cup\left\{e^{i\alpha} : \; \alpha\in [-\alpha_1, \alpha_1] \cup [\pi-\alpha_1, \pi+\alpha_1] \right\} \backslash \{0\},	
\end{align*}
where $\alpha_1 := {\rm arccos}\sqrt{\omega}$.
\end{itemize}
\end{proposition}

\begin{proof}
The spectral problem for the Lax spectrum is obtained by substituting (\ref{constant_1}) into (\ref{KN}),
\begin{equation}
\label{KN-constant-1}
\left\{ \begin{array}{l}
p'(x) = \frac{i}{4}(\lambda^2 - \lambda^{-2}) p 
+ \frac{1}{2} (\lambda + \lambda^{-1}) \sqrt{1-\omega} q, \\
q'(x) = -\frac{i}{4}(\lambda^2-\lambda^{-2})q
-\frac{1}{2} (\lambda + \lambda^{-1}) \sqrt{1-\omega} p. \end{array} \right.
\end{equation}
Looking for nonzero bounded solutions of 
(\ref{KN-constant-1}) in the form 
$p(x) = \hat{p} e^{i \kappa x}$, $q(x) = \hat{q} e^{i \kappa x}$ with $\kappa \in \mathbb{R}$ and constant $(\hat{p},\hat{q}) \in \mathbb{C}^2$, we obtain the characteristic equation in the form
$$
\left| \begin{array}{cc}
	\frac{1}{4}(\lambda^2-\lambda^{-2})-\kappa  & -\frac{i}{2} (\lambda + \lambda^{-1}) \sqrt{1-\omega}   \\
\frac{i}{2} (\lambda + \lambda^{-1}) \sqrt{1-\omega} & -\frac{1}{4}(\lambda^2-\lambda^{-2})-\kappa \\
\end{array} \right| = 0.
$$
Expansion of the determinant yields 
\begin{align*}
	\kappa^2=\frac{1}{16}(\lambda^2-\lambda^{-2})^2+\frac{1}{4}(\lambda+\lambda^{-1})^2 (1-\omega).
\end{align*}
Next we analyze admissible values of $\lambda$ for which $\kappa \in \mathbb{R}$.
\begin{itemize}
	\item For every $\lambda\in\mathbb{R}$ and $\omega \in (-\infty,1)$, we have  $\kappa\in\mathbb{R}$.

\item For  $\lambda\in i\mathbb{R}$, we set $\lambda= i \beta$ with $\beta \in\mathbb{R}$ and obtain 
\begin{align*}
	\kappa = \pm\frac{1}{4}(\beta - \beta^{-1})
	\sqrt{(\beta + \beta^{-1})^2-4(1-\omega)}.
\end{align*}
Since $(\beta + \beta^{-1})^2\geq 4$ for any $\beta \in\mathbb{R}$, 
we have $\kappa \in \R$ for every $\beta \in \R$ if $\omega \in (0,1)$. 
On the other hand, if $\omega \in (-\infty,0]$, then $\kappa \in \R$ 
if and only if $\beta \in (-\infty,-\beta_2] \cup [-\beta_2^{-1},\beta_2^{-1}] \cup [\beta_2,\infty)$, where $\beta_2 := \sqrt{1-\omega} + \sqrt{-\omega}$.

\item For $\lambda \in \mathbb{S}^1$, we set $\lambda = e^{i\alpha}$ with $\alpha \in [0,2\pi]$ and obtain 
\begin{align*}
	\kappa= \pm \cos \alpha \sqrt{\cos^2\alpha-\omega}.
\end{align*}
If $\omega \in (-\infty,0]$, then $\kappa\in \mathbb{R}$ for any $\alpha \in[0,2\pi]$. On the other hand, if $\omega \in (0,1)$, then $\kappa\in \mathbb{R}$ for $\cos^2\alpha \geq\omega$, which 
yields $\alpha \in [-\alpha_1,\alpha_1] \cup [\pi - \alpha_1,\pi + \alpha_1]$ with $\alpha_1 := {\rm arccos}\sqrt{\omega}$.
\end{itemize}
In addition, we checked that $\kappa \notin \R$ if $\lambda \in \mathbb{C} \backslash (\mathbb{R}\cup i \mathbb{R} \cup \mathbb{S}^1)$.
\end{proof}

The end points of the Lax spectrum in Proposition \ref{pro_71} are related to the roots of $P(\lambda)$ in (\ref{polynomial-P}) according to the following corollary. 

\begin{corollary}
	\label{cor-roots-1}
	Let $(U,V)$ be given by (\ref{constant_1}) and consider roots of $P(\lambda)$ in (\ref{polynomial-P}).
	\begin{itemize}
		\item  If $\omega \in (-\infty,0]$, then roots of $P(\lambda)$ are 
		 $\{\pm 1,\pm 1, \pm \beta_2 i,\pm \beta_2^{-1} i\}$,
		where $\beta_2 := \sqrt{1-\omega}+\sqrt{-\omega}$.
		
		\item If $\omega \in (0,1)$, then roots of $P(\lambda)$ are 
			 $\{\pm 1,\pm 1, \pm  e^{i \alpha_1},\pm e^{-i \alpha_1} \}$,
		where $\alpha_1 := {\rm arccos}\sqrt{\omega}$.
	\end{itemize}
\end{corollary}

\begin{proof}
We obtain from (\ref{bb1}) with $b = (1-\omega)^2$ that 
\begin{equation*}
\left\{ \begin{array}{r} 
4 \omega = \lambda_1^2 + \lambda_1^{-2} + \lambda_2^2 + \lambda_2^{-2}, \\
8 \omega - 4 = (\lambda_1^2 + \lambda_1^{-2}) (\lambda_2^2 + \lambda_2^{-2}).
\end{array} \right. 
\end{equation*}	
Eliminating either $\lambda_1^2 + \lambda_1^{-2}$ or $\lambda_2^2 + \lambda_2^{-2}$ yields a quadratic equation with two roots at
$$
\lambda_1^2 + \lambda_1^{-2} = 2, \quad \lambda_2^2 + \lambda_2^{-2} = 4 \omega - 2.
$$
Hence we have $\lambda_1 = 1$ and either $\lambda_2 = i \beta_2$ for $\omega \in (-\infty,0]$ with $\beta_2 := \sqrt{1-\omega}+\sqrt{-\omega}$ or $\lambda_2 = e^{i \alpha_1}$ for $\omega \in (0,1)$ with $\alpha_1 := \arccos(\omega)$.	 
This defines all roots of $P(\lambda)$ in view of the factorization formula (\ref{P-factor}).	
\end{proof}

\begin{remark}
	\label{rem-1}
The part of the Lax spectrum of Proposition \ref{pro_71} on $i \R$ satisfy the stability restriction (\ref{restriction-1}) of Proposition \ref{prop-imag}. However, the other part of the Lax spectrum on $\mathbb{S}^1$ does not satisfy the stability restriction. Hence the linearized MTM system (\ref{line_11}) has the unstable spectrum for the equilibrium point $\mathrm{P}_-$. The end points of the Lax spectrum on either $i\R$ or $\mathbb{S}^1$ are given by roots of $P(\lambda)$ in Corollary \ref{cor-roots-1}. 
\end{remark}

\subsection{Constant-amplitude solution for $\mathrm{P}_+$}

By Proposition \ref{prop-stability-points}, the equilibrium point is a saddle for $\omega \in (-\infty,-1)$, see Figure \ref{fig_w_15} (left). It is located in $\mathbb{T} \times \mathbb{R}_+$ along the blue curve on Figure \ref{fig00}, where $b = (1+\omega)^2$ and $\omega \in (-\infty,-1]$. 
Since $\theta = 0$ and $\xi = -(1+\omega)$ by Proposition \ref{prop-critical-points}, the constant-amplitude solution  (\ref{periodic-waves}) is given by   
\begin{equation}\label{constant_2}
U =\sqrt{-(1+\omega)}, \quad V =\sqrt{-(1+\omega)}, \quad \omega \in (-\infty,-1).
\end{equation}
The Lax spectrum for the admissible solutions of the spectral problem (\ref{KN}) with $(U,V)$ in (\ref{constant_2}) is given by the following proposition.

\begin{proposition}\label{pro_72}
	Let $(U,V)$ be given by (\ref{constant_2}) and consider bounded solutions 
	of the spectral problem (\ref{KN}). If $\omega \in (-\infty,-1)$, then the Lax spectrum is given by 
		\begin{align*}
		\mathbb{R} \cup \left\{ i \beta : \; \beta \in (-\infty, -\beta_2 ]
		\cup [-\beta_2^{-1}, \beta_2^{-1}] \cup 
		[\beta_2,+\infty)  \right\} \backslash \{0\},	
		\end{align*}
		where $\beta_2 := \sqrt{-(1+\omega)} + \sqrt{-\omega}$.
\end{proposition}

\begin{proof}
	The spectral problem for the Lax spectrum is obtained by substituting  (\ref{constant_2}) into (\ref{KN}),
	\begin{equation}
	\label{KN-constant-2}
	\left\{ \begin{array}{l}
	p'(x) = \frac{i}{4}(\lambda^2 - \lambda^{-2}) p 
	- \frac{i}{2} (\lambda - \lambda^{-1}) \sqrt{-(1+\omega)} q, \\
	q'(x) = -\frac{i}{4}(\lambda^2-\lambda^{-2})q
	-\frac{1}{2} (\lambda - \lambda^{-1}) \sqrt{-(1+\omega)} p. 
	\end{array} \right.
	\end{equation}
	Looking for nonzero bounded solutions of 
	(\ref{KN-constant-2}) in the form 
	$p(x) = \hat{p} e^{i \kappa x}$, $q(x) = \hat{q} e^{i \kappa x}$ with $\kappa \in \mathbb{R}$ and constant $(\hat{p},\hat{q}) \in \mathbb{C}^2$, we obtain the characteristic equation in the form
	$$
	\left| \begin{array}{cc}
	\frac{1}{4}(\lambda^2-\lambda^{-2})-\kappa  & -\frac{1}{2} (\lambda - \lambda^{-1}) \sqrt{-(1+\omega)}   \\
	-\frac{1}{2} (\lambda - \lambda^{-1}) \sqrt{-(1+\omega)} & -\frac{1}{4}(\lambda^2-\lambda^{-2})-\kappa \\
	\end{array} \right| = 0.
	$$
Expansion of the determinant yields 	
\begin{align*}
		\kappa^2=\frac{1}{16}(\lambda^2-\lambda^{-2})^2-\frac{1}{4}(\lambda-\lambda^{-1})^2 (1+\omega).
\end{align*}
Next we analyze admissible values of $\lambda$ for which $\kappa \in \mathbb{R}$.
\begin{itemize}
	\item For every $\lambda\in\mathbb{R}$ and $\omega \in (-\infty,-1)$, we have $\kappa\in\mathbb{R}$.
	
	\item For  $\lambda\in i\mathbb{R}$, we set $\lambda= i \beta$ with $\beta \in\mathbb{R}$ and obtain 
	\begin{align*}
	\kappa = \pm\frac{1}{4}(\beta + \beta^{-1})
	\sqrt{(\beta - \beta^{-1})^2+4(1+\omega)}.
	\end{align*}
	Since $\omega \in (-\infty,-1)$, we have $\kappa \in \R$ 
	if and only if $\beta \in (-\infty,-\beta_2] \cup [-\beta_2^{-1},\beta_2^{-1}] \cup [\beta_2,\infty)$, where $\beta_2 := \sqrt{-(1+\omega)} + \sqrt{-\omega}$.
	
	\item For $\lambda \in \mathbb{S}^1$, we set $\lambda = e^{i\alpha}$ with $\alpha \in [0,2\pi]$ and obtain 
	\begin{align*}
	\kappa= \pm i \sin \alpha \sqrt{\cos^2\alpha-(1+\omega)}.
	\end{align*}
	Since $\omega \in (-\infty,-1)$, we have $\kappa \in i \R$ for any $\alpha \in[0,2\pi]$. The points $\alpha = 0$ and $\alpha = \pi$ for which $\kappa = 0$ correspond to $\lambda = \pm 1 \in \R$.
\end{itemize}
In addition, we checked that $\kappa \notin \R$ if $\lambda \in \mathbb{C} \backslash (\mathbb{R}\cup i \mathbb{R} \cup \mathbb{S}^1)$.
\end{proof}

The roots of $P(\lambda)$ in (\ref{polynomial-P}) are given by the following corollary.

\begin{corollary}
	\label{cor-roots-2}
	Let $(U,V)$ be given by (\ref{constant_2}). Roots of $P(\lambda)$ in (\ref{polynomial-P}) are given by 
		$\{\pm i,\pm i, \pm \beta_2 i,\pm \beta_2^{-1} i\}$,
		where $\beta_2 := \sqrt{-(1+\omega)}+\sqrt{-\omega}$.
\end{corollary}
	
\begin{proof}
	We obtain from (\ref{bb1}) with $b = (1+\omega)^2$ that 
	\begin{equation*}
	\left\{ \begin{array}{r} 
	4 \omega = \lambda_1^2 + \lambda_1^{-2} + \lambda_2^2 + \lambda_2^{-2}, \\
	-8 \omega - 4 = (\lambda_1^2 + \lambda_1^{-2}) (\lambda_2^2 + \lambda_2^{-2}).
	\end{array} \right. 
	\end{equation*}	
	Eliminating either $\lambda_1^2 + \lambda_1^{-2}$ or $\lambda_2^2 + \lambda_2^{-2}$ yields a quadratic equation with two roots at
	$$
	\lambda_1^2 + \lambda_1^{-2} = -2, \quad \lambda_2^2 + \lambda_2^{-2} = 4 \omega + 2.
	$$
	Hence we have $\lambda_1 = i$ and $\lambda_2 = i \beta_2$ with $\beta_2 := \sqrt{-(1+\omega)}+\sqrt{-\omega}$.	This defines all the roots of $P(\lambda)$ in view of the factorization formula (\ref{P-factor}).
\end{proof}

\begin{remark}
	\label{rem-2}
	The double roots $\{ \pm i, \pm i \} $ of $P(\lambda)$ in Corollary \ref{cor-roots-2} do not belong to the Lax spectrum of Proposition \ref{pro_72}, whereas the roots $\{ \pm i \beta_2, \pm i \beta_2^{-1}\}$ correspond to the end points of the Lax spectrum on $i \R$. This part of the Lax spectrum satisfy the stability restriction (\ref{restriction-1}) of Proposition \ref{prop-imag}. Hence the linearized MTM system (\ref{line_11}) has the stable spectrum for the equilibrium point $\mathrm{P}_+$.
\end{remark}

\subsection{Numerical approximation of the Lax and stability spectra}

We will now approximate the Lax spectrum numerically in the complex $\lambda$-plane to confirm the conclusions of Propositions \ref{pro_71} and \ref{pro_72}.
Moreover, by using $\Lambda = \pm i \sqrt{P(\lambda)}$, we will 
also plot the stability spectrum of the linearized MTM system (\ref{line_11}) on the complex $\Lambda$-plane.

Figure \ref{fig_6} (a)--(b) gives the Lax and stability spectra for the equilibrium point $\mathrm{P}_-$ with $\omega = -0.1$ obtained from Proposition \ref{pro_71}. The red crosses display the roots of $P(\lambda)$. The numerical approximations of the Lax and stability spectra are shown in Figure \ref{fig_6} (c)-(d) for the same value $\omega = -0.1$. Details of the numerical method are explained in Appendix \ref{app_1}.

The same results are shown in Figure \ref{fig_7} for the equilibrium point $\mathrm{P}_-$ with $\omega=0.1$. Compared with Figure \ref{fig_6} and in agreement with Proposition \ref{pro_71}, the Lax spectrum on $i \R$ has gaps for $\omega = -0.1$ and no gaps for $\omega = 0.1$ and the Lax spectrum on $\mathbb{S}^1$ has gaps for $\omega = 0.1$ and no gaps for $\omega = -0.1$.

\begin{figure}[htb!]
	\centering
	\subfigure[Lax spectrum in $\lambda$-plane. ] {\includegraphics[width=2.1in,height=1.3in]{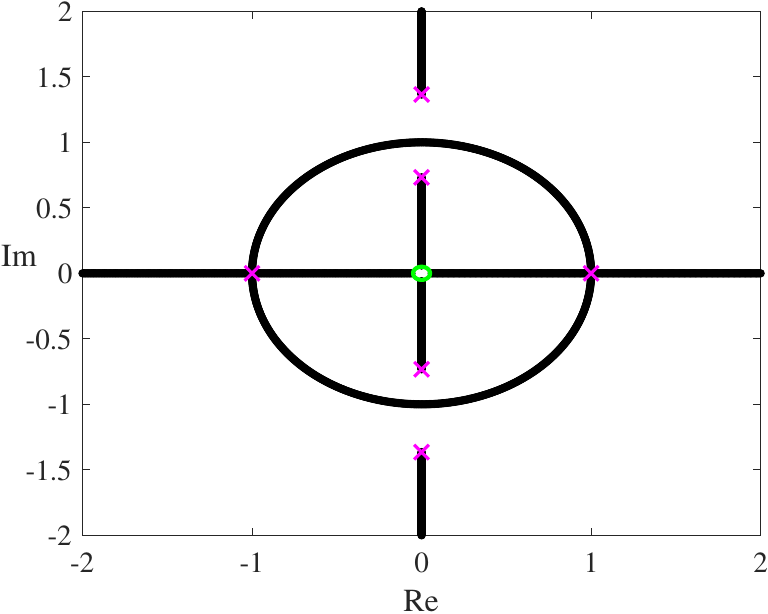}}
	\subfigure[Stability spectrum in $\Lambda$-plane.]
	{\includegraphics[width=2.1in,height=1.3in]{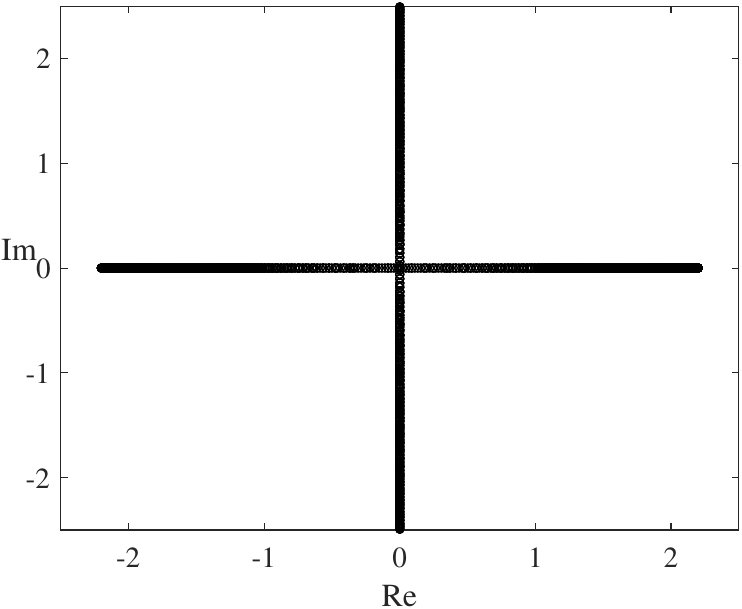}}\\
	\subfigure[Lax spectrum in $\lambda$-plane. ] {\includegraphics[width=2.1in,height=1.3in]{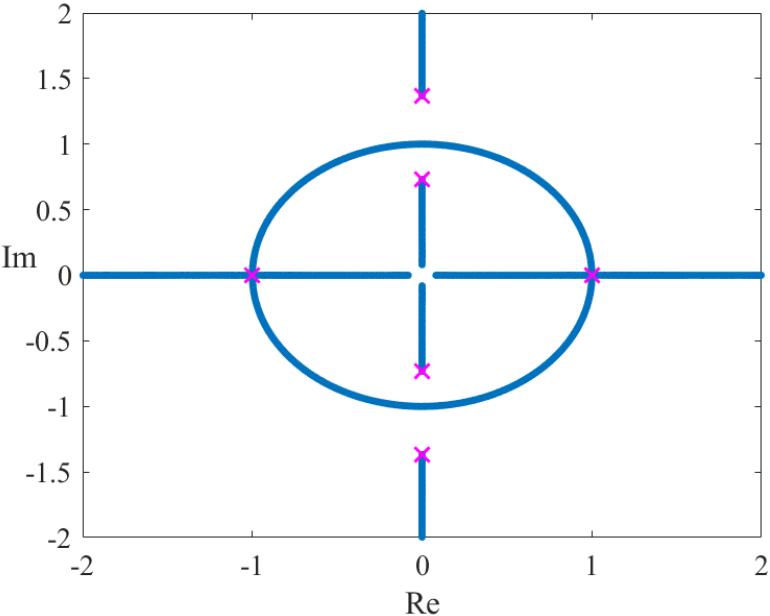}}
	\subfigure[Stability spectrum in  $\Lambda$-plane.]
	{\includegraphics[width=2.1in,height=1.3in]{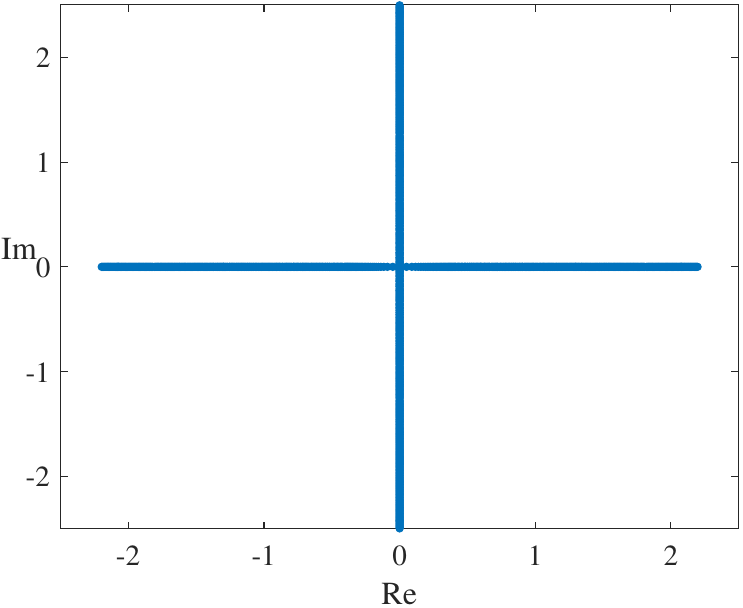}}
	\caption{Lax and stability spectra for the equilibrium point $\mathrm{P}_-$ with $\omega=-0.1$ obtained analytically (a)--(b) and numerically (c)--(d). }\label{fig_6}
\end{figure}

\begin{figure}[htb!]
	\centering
	\subfigure[Lax spectrum in $\lambda$-plane.] {\includegraphics[width=2.1in,height=1.3in]{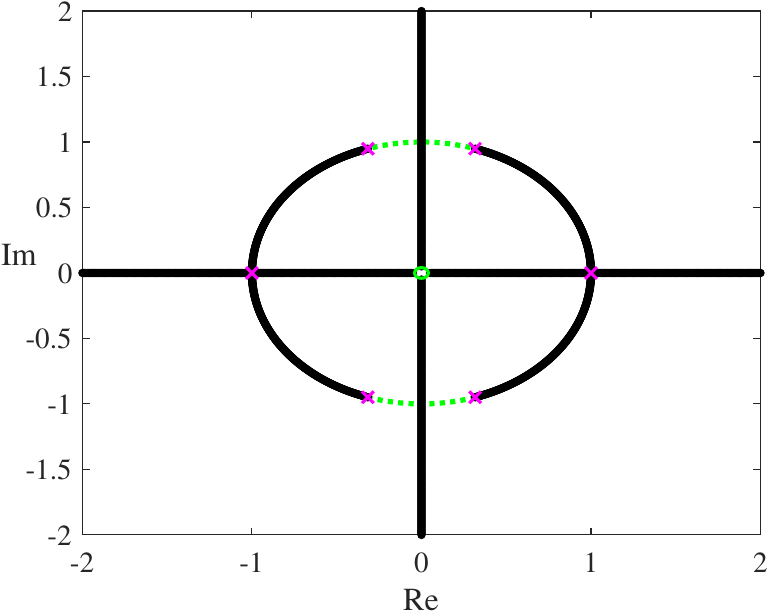}}
	\subfigure[Stability spectrum in $\Lambda$-plane.]
	{\includegraphics[width=2.1in,height=1.3in]{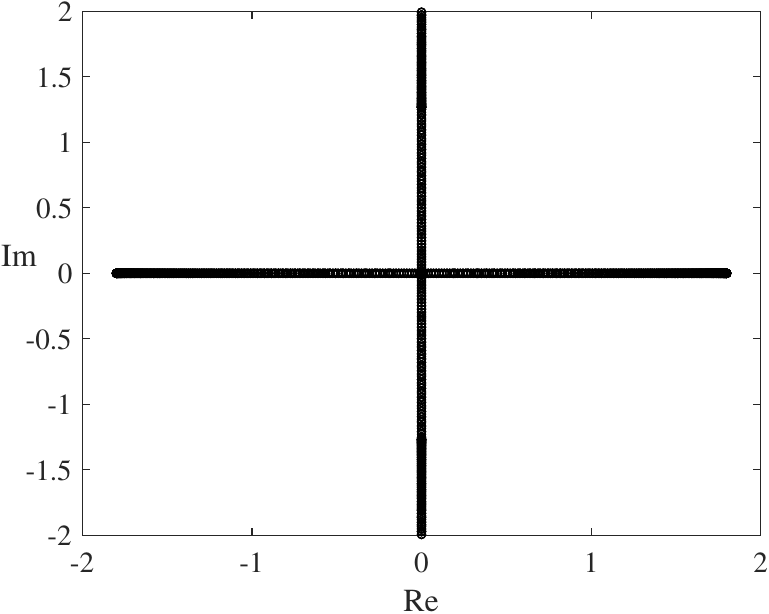}}\\
	\subfigure[Lax spectrum in $\lambda$-plane. ] {\includegraphics[width=2.1in,height=1.3in]{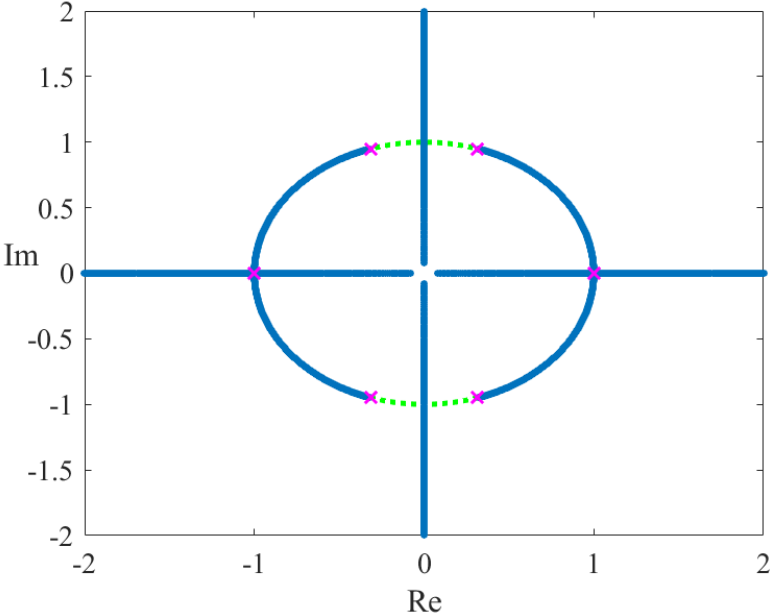}}
	\subfigure[Stability spectrum in  $\Lambda$-plane.] {\includegraphics[width=2.1in,height=1.3in]{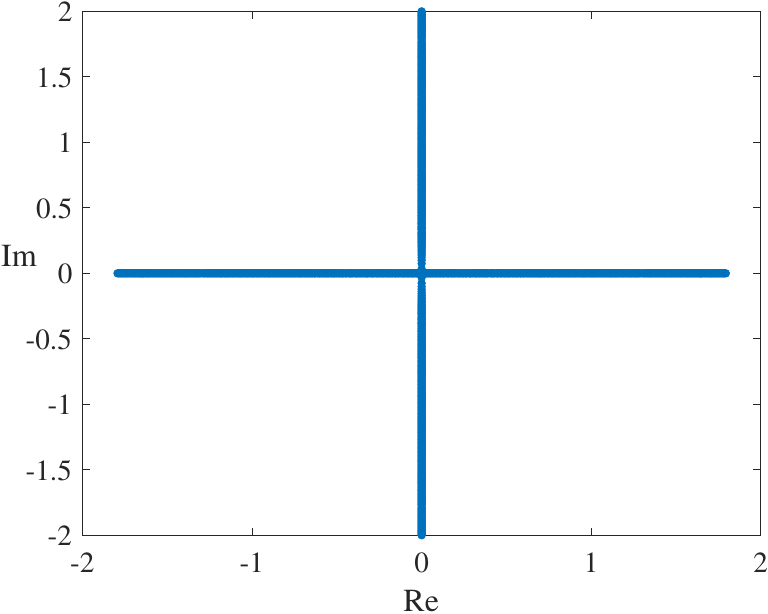}}
	\caption{The same as in Figure \ref{fig_6} but for $\mathrm{P}_-$ with $\omega=0.1$. } \label{fig_7}
\end{figure}

Figure \ref{fig_8} (a)--(b) gives the Lax and stability spectra for the equilibrium point $\mathrm{P}_+$ with $\omega = -1.2$ obtained from Proposition \ref{pro_72}. The red crosses display again the roots of $P(\lambda)$. Note that 
the double roots $\pm i$ do not belong to the Lax spectrum. 
The numerical approximations of the Lax and stability spectra 
are shown in Figure \ref{fig_8} (c)-(d) for the same value $\omega = -1.2$. The green dotted curve shows the unit circle $\mathbb{S}^1$ which is not a part of the Lax spectrum. We have again a full agreement between the theory and the numerical approximations. 

\begin{figure}[htb!]
	\centering
	\subfigure[Lax spectrum in $\lambda$-plane. ] {\includegraphics[width=2.1in,height=1.6in]{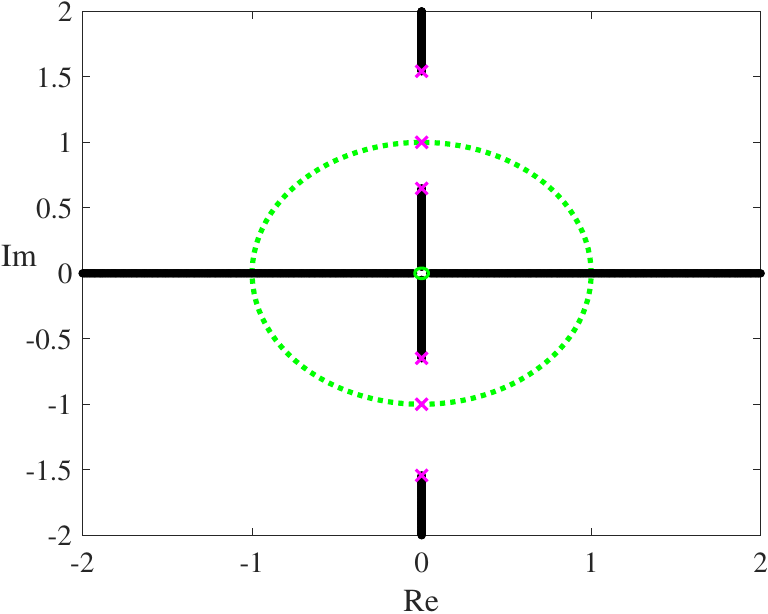}}
	\subfigure[Stability spectrum in $\Lambda$-plane.] {\includegraphics[width=2.1in,height=1.6in]{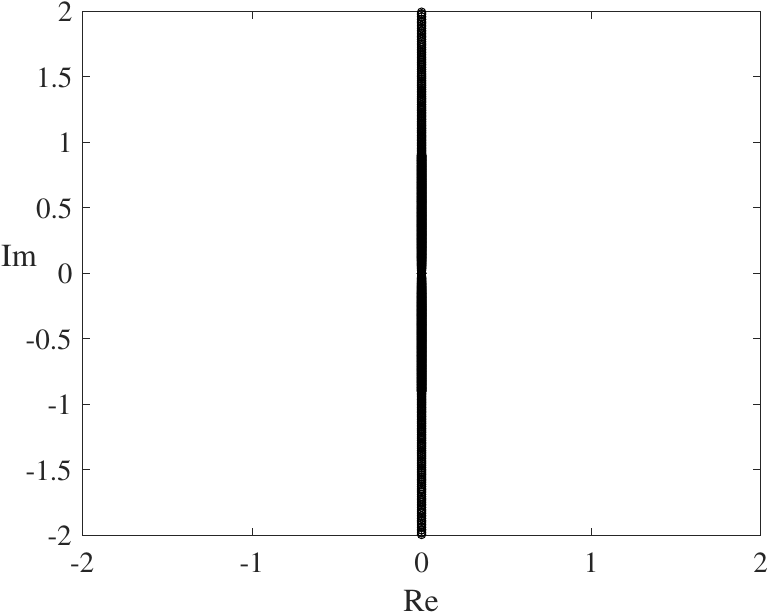}}\\
	\subfigure[Lax spectrum in $\lambda$-plane. ] {\includegraphics[width=2.1in,height=1.6in]{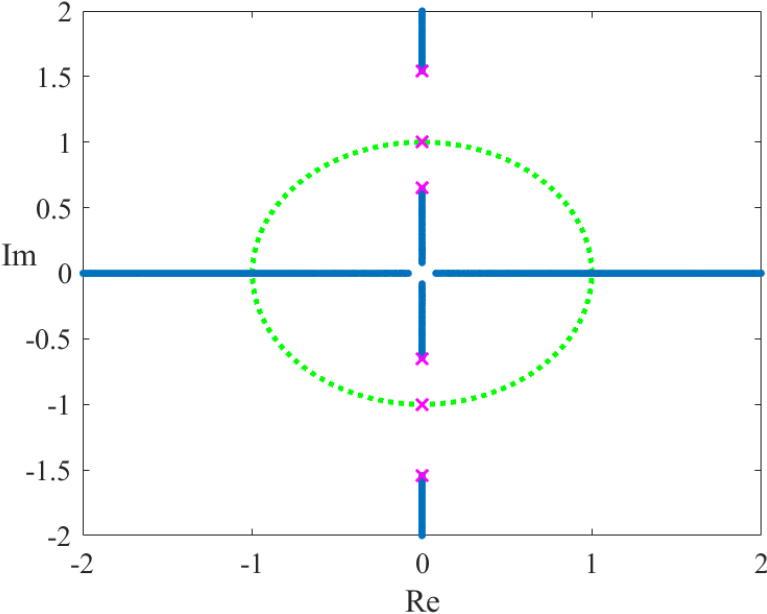}}
	\subfigure[Stability spectrum in $\Lambda$-plane.] {\includegraphics[width=2.1in,height=1.6in]{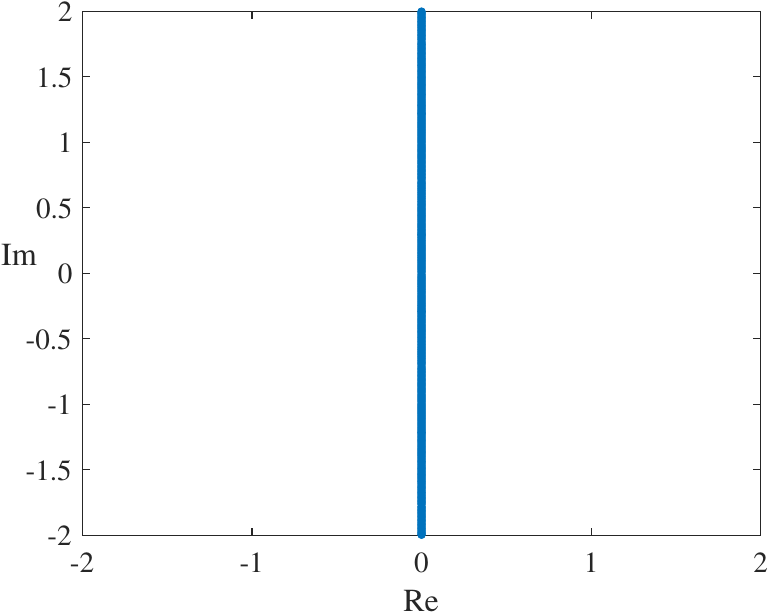}}
	\caption{The same as in Figure \ref{fig_6} but for $\mathrm{P}_+$ with $\omega=-1.2$. } \label{fig_8}
\end{figure}

\begin{remark}
	In agreement with Remarks \ref{rem-1} and \ref{rem-2}, Figures \ref{fig_6}, \ref{fig_7} and \ref{fig_8} confirm that 
	the center equilibrium point $\mathrm{P}_-$ is spectrally unstable 
	and the saddle equilibrium point $\mathrm{P}_+$ is spectrally stable in the spectral stability problem (\ref{line_11}). 
\end{remark}

\section{Lax and stability spectra for standing periodic waves}
\label{sec_periodic_1}

Here we approximate  the Lax and stability spectra for the standing periodic waves (\ref{standing-wave}) with $V = \bar{U}$. By using (\ref{periodic-waves}), (\ref{hamiton_2}), and (\ref{re_eq_2}) with $b = -2\omega\xi-\xi^2-2\xi\cos\theta$, the profile $\xi$ of the standing 
periodic waves is obtained from the first-order invariant
\begin{equation}
\label{exact_1}
(\xi')^2+R(\xi)=0,
\end{equation}
where 
\begin{align*}
R(\xi) &=(b+2\omega\xi+\xi^2)^2-4\xi^2 \\
&=\xi^4+4\omega\xi^3+(4\omega^2+2b-4)\xi^2+4b\omega\xi+b^2 \\
&=(\xi-u_1)(\xi-u_2)(\xi-u_3)(\xi-u_4),
\end{align*}
where $(u_1,u_2,u_3,u_4)$ satisfy
	\begin{equation}\label{521}
	\begin{cases}
	u_1+u_2+u_3+u_4=-4\omega,\\
	u_1u_2+u_1u_3+u_1u_4+u_2u_3+u_2u_4+u_3u_4=4\omega^2+2b-4,\\
	u_1u_2u_3+u_1u_2u_4+u_1u_3u_4+u_2u_3u_4=-4b\omega,\\
	u_1u_2u_3u_4=b^2.
	\end{cases}	
	\end{equation}
The following proposition related roots $\{ u_1,u_2,u_3,u_4 \}$ of the quartic polynomial $R(\xi)$ in (\ref{exact_1}) to roots $\{ \pm \lambda_1, \pm \lambda_2, \pm \lambda_1^{-1}, \pm \lambda_2^{-2}\}$ of $P(\lambda)$ in (\ref{polynomial-P}). For the cubic and derivative NLS equations, such relations were found by Kamchatnov \cite{Kam1,Kam2}.

\begin{proposition}\label{pro_u1}
	Roots $(u_1,u_2,u_3,u_4)$ of $R(\xi)$ are related to roots $\{ \pm \lambda_1, \pm \lambda_2, \pm \lambda_1^{-1}, \pm \lambda_2^{-1}\}$ of $P(\lambda)$ by
	\begin{equation}
	\begin{cases}
	u_1=-\frac{1}{4}(\lambda_{1}-\lambda_1^{-1}-\lambda_2+\lambda_2^{-1})^2,\\
	u_2=-\frac{1}{4}(\lambda_{1}+\lambda_1^{-1}-\lambda_2-\lambda_2^{-1})^2,\\
	u_3=-\frac{1}{4}(\lambda_{1}-\lambda_1^{-1}+\lambda_2-\lambda_2^{-1})^2,\\
	u_4=-\frac{1}{4}(\lambda_{1}+\lambda_1^{-1}+\lambda_2+\lambda_2^{-1})^2.\\
	\end{cases}	
	\label{u1-u4}
	\end{equation}
\end{proposition}

\begin{proof}
	The proof is a direct calculation. We show that relations (\ref{bb1}) and (\ref{u1-u4}) recover relations (\ref{521}). For the first two equations of system (\ref{521}), we obtain 
	\begin{align*}
u_1+u_2+u_3+u_4 &= -(\lambda_1^2+\lambda_1^{-2}+\lambda_2^2+\lambda_2^{-2}) = -4\omega
\end{align*}
and 
\begin{align*}
& \qquad u_1u_2+u_1u_3+u_1u_4+u_2u_3+u_2u_4+u_3u_4 \\
&=\frac{3}{8}(\lambda_1^2 + \lambda_1^{-2} + \lambda_2^2 + \lambda_1^{-2})^2
	- \frac{1}{2} (\lambda_1^2 + \lambda_1^{-2})
	(\lambda_2^2+\lambda_2^{-2}) -4\\
	&= 4\omega^2+2b-4.
\end{align*}
For the last two equations of system (\ref{521}), we obtain from 
 (\ref{bb1}) that 
	\begin{align*}
	4b &= \frac{1}{4}(\lambda_1^2 +\lambda_1^{-2} +\lambda_2^2 +\lambda_2^{-2})^2 - (\lambda_1^2 + \lambda_1^{-2})(\lambda_2^2+\lambda_2^{-2}) \\
	&= \frac{(\lambda_1\lambda_2-1)^2(\lambda_1\lambda_2+1)^2(\lambda_1-\lambda_2)^2(\lambda_1+\lambda_2)^2}{4\lambda_1^4\lambda_2^4},\\
	\end{align*}
which yields
	\begin{equation*}
	b^2=\frac{(\lambda_1\lambda_2-1)^4(\lambda_1\lambda_2+1)^4(\lambda_1-\lambda_2)^4(\lambda_1+\lambda_2)^4}{256\lambda_1^8\lambda_2^8},
	\end{equation*}
	and
	\begin{equation*}
	4b\omega
	=\frac{(\lambda_1\lambda_2-1)^2(\lambda_1\lambda_2+1)^2(\lambda_1-\lambda_2)^2(\lambda_1+\lambda_2)^2(\lambda_1^2+\lambda_2^2)(\lambda_1^2\lambda_2^2+1)}{16\lambda_1^6\lambda_2^6}.
	\end{equation*}
These two equations define the right-hand sides of the last two equations of system (\ref{521}). Substituting (\ref{u1-u4}) to the left-hand sides of the last two equations of system (\ref{521}) and using the symbolic software program MAPLE, we verify their equivalence with the right-hand sides. 
\end{proof}

Next, we identify the possible roots $\{ \pm \lambda_1, \pm \lambda_2, \pm \lambda_1^{-1}, \pm \lambda_2^{-1}\}$ of $P(\lambda)$ in each region 
of the parameter plane $(b,\omega)$, see Figure \ref{fig00}, where the standing periodic waves exist.

\begin{proposition}
	\label{prop-roots}
Recall that $\{ \pm \lambda_1, \pm \lambda_2, \pm \lambda_1^{-1}, \pm \lambda_2^{-1}\}$ are roots of $P(\lambda)$ given by (\ref{polynomial-P}). Then, we have 
	\begin{itemize}
		\item $\lambda_1 = \bar{\lambda}_2 \in \mathbb{C} \backslash ( \mathbb{R} \cup i\mathbb{R} \cup \mathbb{S}^1 )$ in region {\rm \Rmnum{1}},
		\item $\lambda_1,\lambda_2 \in \mathbb{S}^1 \backslash ( \mathbb{R} \cup i\mathbb{R})$ with $\lambda_1 \neq \lambda_2$ in region $\rm \Rmnum{2}_A$,
		\item $\lambda_1 \in \mathbb{S}^1 \backslash ( \mathbb{R} \cup i\mathbb{R})$ and $\lambda_2 = i \beta_2 \in i \R \backslash (\mathbb{R} \cup \mathbb{S}^1 )$ in region $\rm \Rmnum{2}_B$,
		\item $\lambda_1 = i\beta_1 \in i \R \backslash (\mathbb{R} \cup \mathbb{S}^1 )$ and $\lambda_2 = i \beta_2 \in i \R \backslash (\mathbb{R} \cup \mathbb{S}^1 )$  with $\beta_1 \neq \beta_2$ in region {\rm \Rmnum{3}}.
	\end{itemize}
\end{proposition}

\begin{proof}
Solving $P(\lambda)= 0$ in (\ref{polynomial-P}) yields
	\begin{equation}\label{p_roots}
	\lambda^2+\frac{1}{\lambda^2}=2\omega\pm2\sqrt{b}.	
	\end{equation}
Since $\{ \pm \lambda_1, \pm \lambda_2, \pm \lambda_1^{-1}, \pm \lambda_2^{-2}\}$ are roots of $P(\lambda)$, we have 
	\begin{equation}\label{416}
	\left\{ 
	\begin{aligned}
	\lambda_1^2    &= (\omega+\sqrt{b})+\sqrt{(\omega+\sqrt{b})^2-1},\\
	\lambda_1^{-2} &= (\omega+\sqrt{b})-\sqrt{(\omega+\sqrt{b})^2-1},\\
	\lambda_2^2    &= (\omega-\sqrt{b})+\sqrt{(\omega-\sqrt{b})^2-1},\\
	\lambda_2^{-2} &= (\omega-\sqrt{b})-\sqrt{(\omega-\sqrt{b})^2-1}.
	\end{aligned}\right.
	\end{equation}
This allows us to consider different cases in each existence region 
of the parameter plane $(b,\omega)$ given by Proposition \ref{pro_40}.
	
\begin{itemize}
	\item Since $b \in (-\infty,0)$ and $\omega \in \mathbb{R}$ in region \Rmnum{1}, we obtain from (\ref{416}) that $\lambda_1 = \bar{\lambda}_2 \in \mathbb{C} \backslash ( \mathbb{R} \cup i\mathbb{R} \cup \mathbb{S}^1 )$.\\
	
	\item We have $b \in \big( 0, \min\{(1-\omega)^2, (1+\omega)^2\} \big)$ and  $\omega \in (-1,1)$ in region $\rm \Rmnum{2}_A$. It is sufficient to consider the case of $\omega \in [0,1)$ for which  $b \in (0,(1-\omega)^2)$, since the case of $\omega \in (-1,0)$ is similar. Since $\omega+\sqrt{b}, \omega-\sqrt{b} \in (-1,1)$, we obtain from (\ref{416}) that $\lambda_1,\lambda_2 \in \mathbb{S}^1 \backslash ( \mathbb{R} \cup i\mathbb{R})$ with $\lambda_1 \neq \lambda_2$. \\
	
	\item We have $b \in \big((1+\omega)^2,(1-\omega)^2\big)$ and $\omega\in(-\infty,0)$ in region $\rm \Rmnum{2}_B$. Since $\omega+\sqrt{b}\in(-1,1)$ and 
	$\omega-\sqrt{b}<-1$, we obtain from (\ref{416}) that $\lambda_1 \in \mathbb{S}^1 \backslash ( \mathbb{R} \cup i\mathbb{R})$ and $\lambda_2 = i \beta_2 \in i \R \backslash (\mathbb{R} \cup \mathbb{S}^1 )$. \\
	
	\item We have $b \in (0,(1+\omega)^2)$ with $\omega \in (-\infty,-1)$ 
	in region \Rmnum{3}. Since $\omega+\sqrt{b}<-1$ and $\omega-\sqrt{b}<-1$, we obtain from (\ref{416}) that $\lambda_1 = i\beta_1 \in i \R \backslash (\mathbb{R} \cup \mathbb{S}^1 )$ and $\lambda_2 = i \beta_2 \in i \R \backslash (\mathbb{R} \cup \mathbb{S}^1 )$  with $\beta_1 \neq \beta_2$. 
\end{itemize}	
This concludes the analysis since the standing periodic waves do not exist in region  \Rmnum{4}.
\end{proof}

We use Propositions \ref{pro_u1} and \ref{prop-roots} in order to investigate the spectral stability of the standing periodic waves in different regions of the parameter plane $(b,\omega)$, see Figure \ref{fig00}. In each region, we give the explicit representation for the periodic solution $\xi(x)$ and the roots of $P(\lambda)$, after which we select several sample points to approximate the Lax spectrum numerically according to the method explained in Appendix \ref{app_1}. The sample points in each region are shown in Figure \ref{area6}. From the numerical approximations of the Lax spectrum, we find the stability spectrum by using (\ref{Lambda-Omega}). This will verify the stability and instability conclusions shown in Figure \ref{fig_total}. 

\begin{figure}[H]
	\centering
	\includegraphics[width=3.6in,height=2.5in]{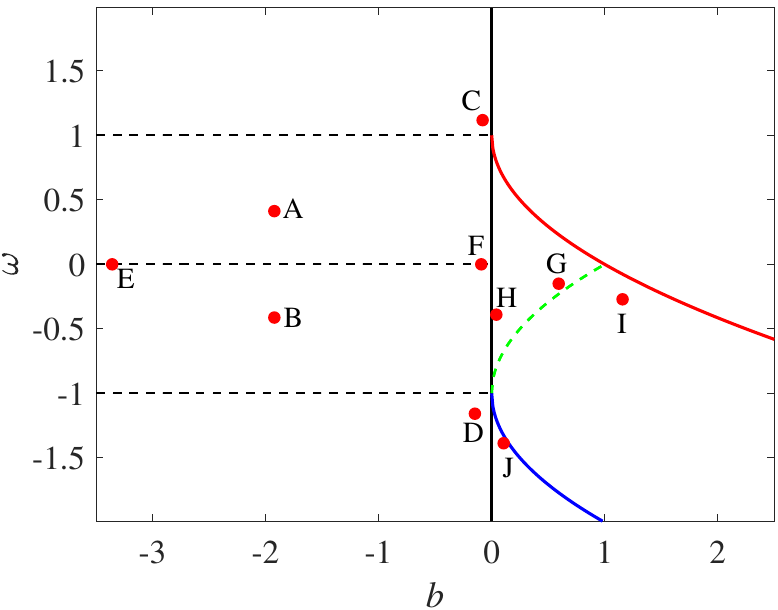}
	\caption{Sample points in each region on the $(b,\omega)$ plane for numerical approximations.}
	\label{area6}
\end{figure}

\subsection{Standing periodic waves in region \Rmnum{1}}

By Proposition \ref{prop-roots}, roots of $P(\lambda)$ form two complex quadruplets, which are reflected symmetrically relative to the unit circle $\mathbb{S}^1$.  
For definiteness, we write $\lambda_1=\alpha_1+i\beta_1$ and 
$\lambda_2= \alpha_1 -i\beta_1$ with $\alpha_1^2 + \beta_1^2 \neq 1$ 
so that equations (\ref{521}) yield
\begin{equation*}
\omega =\frac{(\alpha_1^2-\beta_1^2)((\alpha_1^2+\beta_1^2)^2 + 1)}{2(\alpha_1^2 +\beta_1^2)^2}
\end{equation*}
and
\begin{equation*}
b=-\frac{(\alpha_1^2+\beta_1^2-1)^2(\alpha_1^2+\beta_1^2+1)^2\beta_1^2\alpha_1^2}{(\alpha_1^2+\beta_1^2)^4}.
\end{equation*}
By using (\ref{u1-u4}), we obtain 
\begin{equation*}
\begin{cases}
u_1= \frac{\beta_1^2 (\alpha_1^2+\beta_1^2+1)^2}{(\alpha_1^2+\beta_1^2)^2},\\
u_2=\frac{\beta_1^2 (\alpha_1^2+\beta_1^2-1)^2}{(\alpha_1^2+\beta_1^2)^2},\\
u_3=-\frac{\alpha_1^2 (\alpha_1^2+\beta_1^2-1)^2}{(\alpha_1^2+\beta_1^2)^2},\\
u_4=-\frac{\alpha_1^2 (\alpha_1^2+\beta_1^2+1)^2}{(\alpha_1^2+\beta_1^2)^2},
\end{cases}
\end{equation*}
which satisfy the ordering $u_4\leq u_3\leq 0\leq u_2\leq u_1$. The exact  periodic solution of the first-order invariant (\ref{exact_1}) with this ordering can be written in the explicit form, see \cite{CP_1}:
\begin{equation}\label{solution_1}
\xi(x)=u_4+\frac{(u_1-u_4)(u_2-u_4)}{(u_2-u_4)+(u_1-u_2){\rm sn}^2(\nu x;k)},
\end{equation}
where 
$$
\nu=\frac{1}{2}\sqrt{(u_1-u_3)(u_2-u_4)}, \quad 
k=\frac{\sqrt{(u_1-u_2)(u_3-u_4)}}{\sqrt{(u_1-u_3)(u_2-u_4)}}.
$$
The periodic solution $\xi(x)$ in (\ref{solution_1}) is located in the interval $[u_2,u_1]$ and has period $L=2K(k)\nu^{-1}$. The component $\theta(x)$ 
in the standing periodic waves (\ref{standing-wave}) and (\ref{periodic-waves}) is given by 
\begin{equation}
\label{theta}
\theta(x)=	\left\{ \begin{array}{ll}
\arccos{ \frac{-(\xi^2+2\omega\xi+b)}{2\xi}}, \quad & 0\leq\theta\leq\pi,\\
2\pi-\arccos{ \frac{-(\xi^2+2\omega\xi+b)}{2\xi}}, \quad & \pi \leq \theta \leq 2\pi.
\end{array} \right.
\end{equation}
Since 
$$
\theta'= \xi - b \xi^{-1} > 0
$$
follows from (\ref{re_eq_2}) with $b =-2\omega \xi-\xi^2 -2 \xi \cos \theta$,
the mapping $x \mapsto \theta(x)$ is monotonically increasing in each period, in agreement with the phase portraits on Figures \ref{fig_w_15}, \ref{fig_w_1}, and \ref{fig_w_16} for $b < 0$.

\begin{remark}
As $b \to 0$, the point inside the existence region {\rm \Rmnum{1}} approaches to the boundary given by the black vertical line on Figure \ref{fig00}. At this boundary, bifurcations of the two complex quadruplets 
in the roots of $P(\lambda)$ depend on the value of $\omega \in \R$. 
\begin{itemize}
	\item If $\omega \in (-1,1)$, each of the four pairs of roots of $P(\lambda)$ coalesce on $\mathbb{S}^1$. The solution continued inside region $\rm \Rmnum{2}_A$ has two complex quadruplets in the roots of $P(\lambda)$ on $\mathbb{S}^1$. \\
	
	\item  If $\omega \in (1,\infty)$, each of the four pairs of roots 
of $P(\lambda)$ coalesce on $\mathbb{R}$. The four double roots on $\R$ are reflected symmetrically relatively to the values $\pm 1$. The solution cannot be continued inside region {\rm \Rmnum{4}}. \\

\item If $\omega \in (-\infty,-1)$, each of the four pairs of roots 
of $P(\lambda)$ coalesce on $i \mathbb{R}$. The four double roots on $i \R$ are reflected  symmetrically relatively to the values $\pm i$. The solution continued inside region {\rm \Rmnum{3}} has four pairs of roots on $i \R$ which are reflected symmetrically relatively to the values $\pm i$.
\end{itemize}
\end{remark}

To compute the Lax and stability spectra, we use the numerical method from 
Appendix \ref{app_1} and approximate the Floquet spectrum at different points 
in region \Rmnum{1} shown in Figure \ref{area6} after which we use the transformation $\Lambda=\pm{i}\sqrt{P(\lambda)}$ to approximate the stability spectrum of the standing periodic waves. 

At point A, we take $\alpha_1 = 1.4$ and $\beta_1=1.1$. The Lax spectrum is shown in Figure \ref{fig_13}(a), where the red crosses represent roots of $P(\lambda)$ as two complex quadruplets, which are reflected symmetrically about $\mathbb{S}^1$. The stability spectrum is shown in Figure \ref{fig_13}(b) and contains the figure-eight instability band. 

Figure \ref{fig_133} shows similar Lax and stability spectra at point B, 
for which we take $\alpha_1=1.1$ and $\beta_1=1.4$. The Lax spectra between the two cases are only different by the convexity of the spectral bands connecting the complex quadruplets, whereas the stability spectra are very similar and contain the figure-eight instability band. 

\begin{figure}[htb!]
	\centering
	\subfigure[Lax spectrum in $\lambda$-plane.]
	{\includegraphics[width=2.1in,height=1.9in]{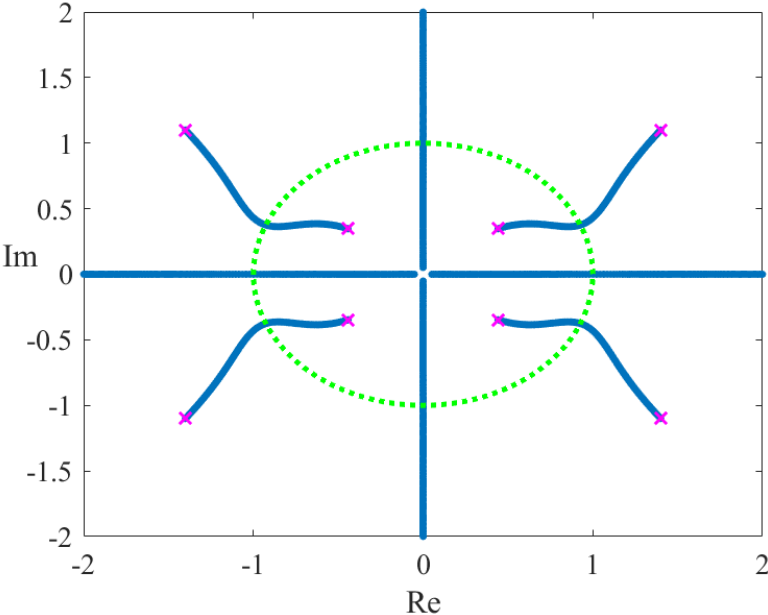}}
	\subfigure[Stability spectrum in $\Lambda$-plane.]
	{\includegraphics[width=2.1in,height=1.9in]{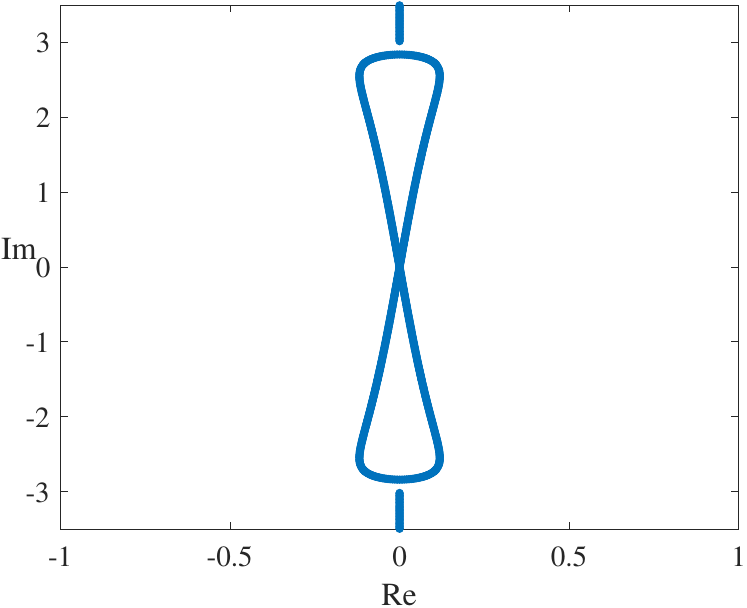}}
	\caption{The Lax and stability spectra for the standing periodic wave  with  $\alpha_1=1.4$ and $\beta_1=1.1$ at point A.}
	\label{fig_13}
\end{figure}

\begin{figure}[htb!]
	\centering
	\subfigure[Lax spectrum in $\lambda$-plane.] {\includegraphics[width=2.1in,height=1.9in]{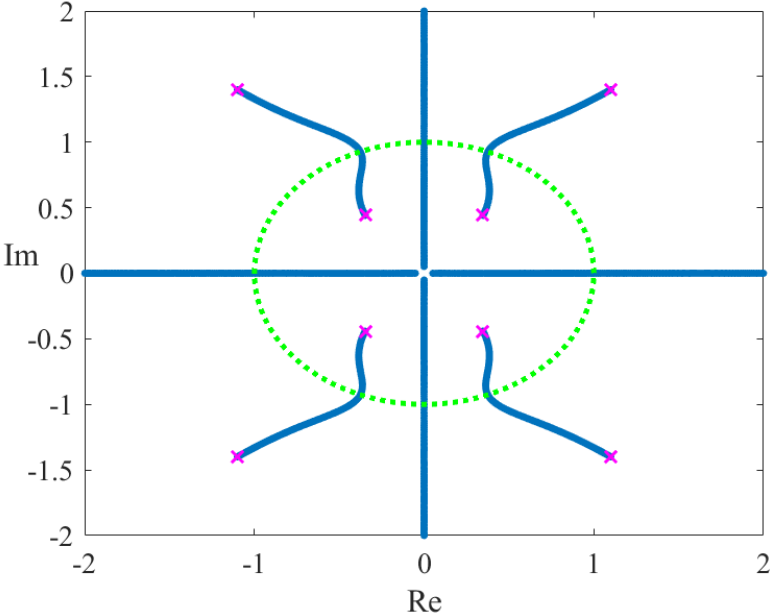}}
	\subfigure[Stability spectrum in $\Lambda$-plane.]
	{\includegraphics[width=2.1in,height=1.9in]{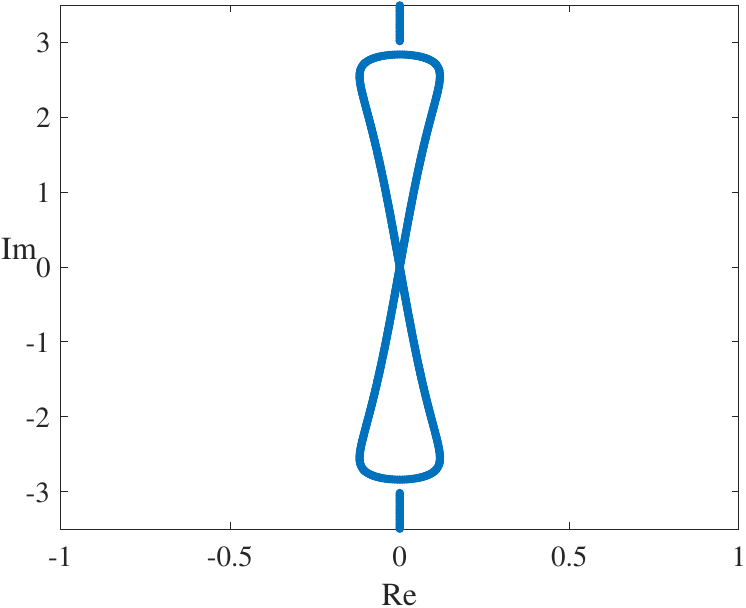}}
	\caption{The same as Figure \ref{fig_13} but for $\alpha_1=1.1$ and $\beta_1=1.4$ at point B.}\label{fig_133}
\end{figure}

Figure \ref{fig_14} shows the Lax and stability spectra for point $C$ close to the boundary $b = 0$ for $\omega \in (1,\infty)$, for which we take $\alpha_1=1.1$ and $\beta_1=0.1$. The complex quadruplets are very close to the real axis of the $\lambda$ plane. The same figure-eight appears in the stability spectrum in the $\Lambda$ plane. Figure \ref{fig_15} shows the Lax and stability spectra for point $D$ which is also close to the boundary $b = 0$ but for $\omega \in (-\infty,-1)$, for which we take $\alpha_1=0.2$ and $\beta_1=1.2$. The complex quadruplets are now close to the imaginery axis of the $\lambda$ plane. We have checked that the spectral bands of the Lax spectrum outside $\mathbb{S}^1$ are transformed to one half of the figure-eight in the stability spectrum located for ${\rm Re}(\Lambda) {\rm Im}(\Lambda) < 0$, whereas the spectral bands of the Lax spectrum inside  $\mathbb{S}^1$ are transformed to the other half of the figure-eight in the stability spectrum located for ${\rm Re}(\Lambda) {\rm Im}(\Lambda)> 0$.

\begin{figure}[htb!]
	\centering
	\subfigure[Lax spectrum in $\lambda$-plane.]
	{\includegraphics[width=2.1in,height=1.9in]{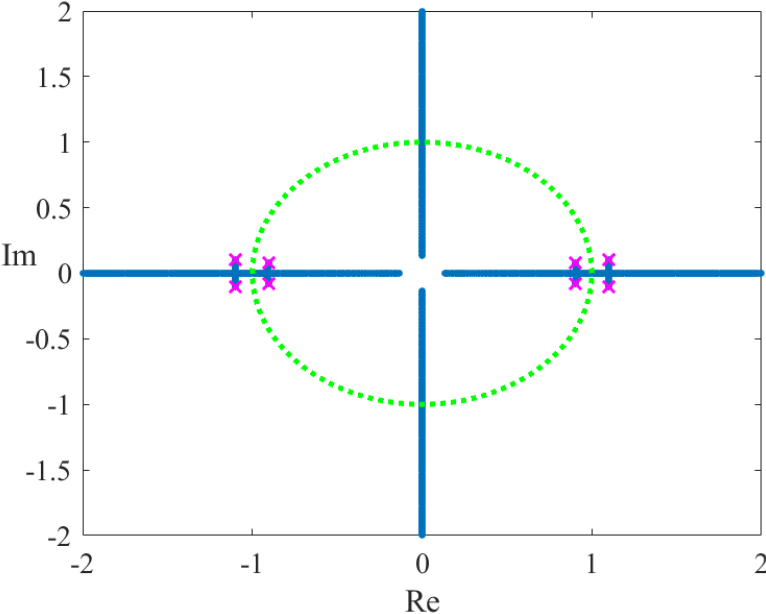}}
	\subfigure[Stability spectrum in $\Lambda$-plane.]
	{\includegraphics[width=2.1in,height=1.9in]{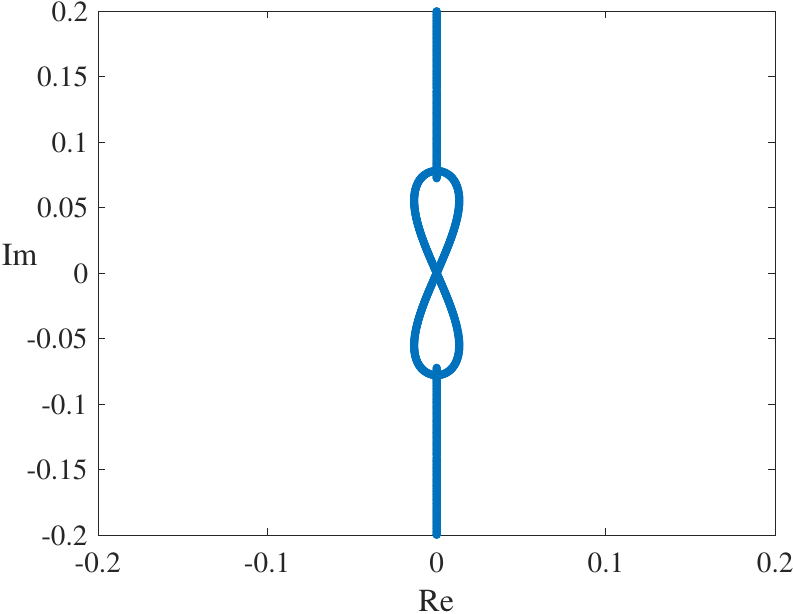}}
	\caption{The same as Figure \ref{fig_13} but for $\alpha_1=1.1$ and $\beta_1=0.1$ at point C. }\label{fig_14}
\end{figure}

\begin{figure}[htb!]
	\centering
	\subfigure[Lax spectrum in $\lambda$-plane.]
	{\includegraphics[width=2.1in,height=1.9in]{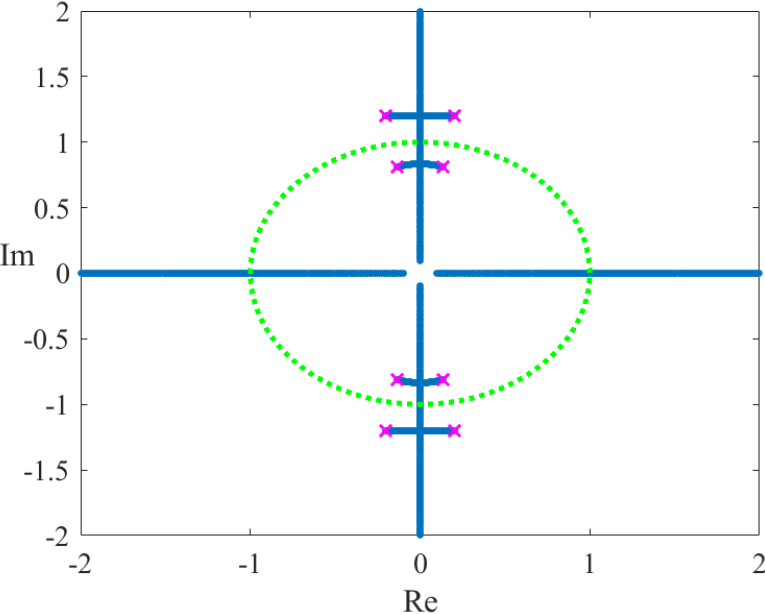}}
	\subfigure[Stability spectrum in $\Lambda$-plane.]
	{\includegraphics[width=2.1in,height=1.9in]{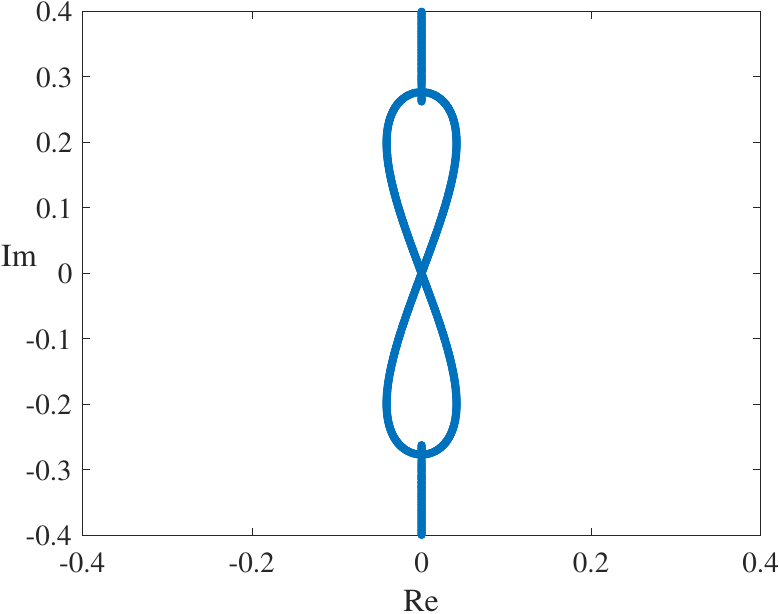}}
	\caption{The same as Figure \ref{fig_13} but for  $\alpha_1=0.2$ and $\beta_1=1.2$ at point D. }\label{fig_15}
\end{figure}

In the symmetric case $\omega = 0$, the roots of $P(\lambda)$ are located at the diagonals in the $\lambda$ plane. Figure \ref{fig_16} shows the Lax and stability spectra at point E for which  $\alpha_1=\beta_1=1.4$. We can see that the spectral bands of the Lax spectrum are located at the diagonals in the $\lambda$ plane. Consequently, the stability spectrum is located on $i \R$ and the figure-eight shrinks to the vertical line. We note that the derivative NLS equation considered in \cite{CPU_1} does not have such families of standing periodic wave solutions. 

Figure \ref{fig_17} shows the same at point F for which  $\alpha_1=\beta_1=0.8$. Since the point is close to the boundary $b = 0$ for $\omega \in (-1,1)$, the complex quadruplets are  close to the unit circle $\mathbb{S}^1$ in the $\lambda$ plane, whereas the stability spectrum on $i \R$ admits some bandgaps.

\begin{figure}[htb!]
	\centering
	\subfigure[Lax spectrum in $\lambda$-plane.] {\includegraphics[width=2.1in,height=1.9in]{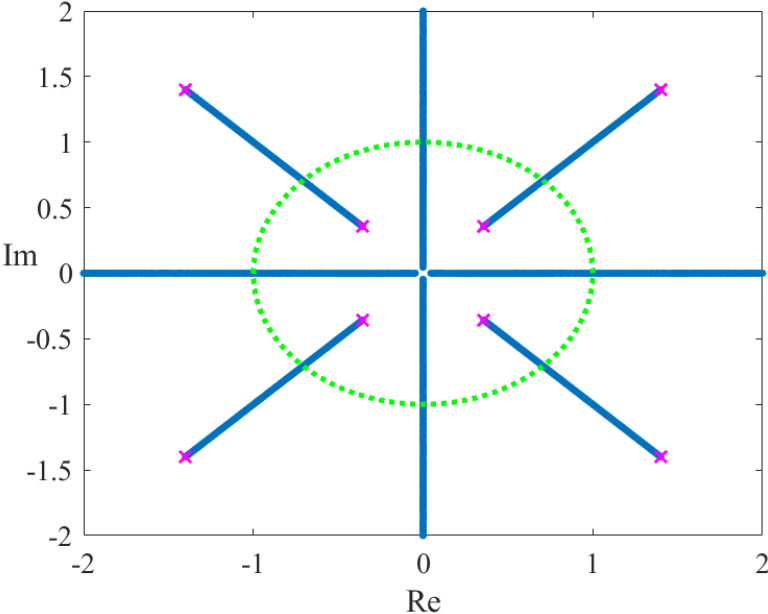}}
	\subfigure[Stability spectrum in $\Lambda$-plane.]
	{\includegraphics[width=2.1in,height=1.9in]{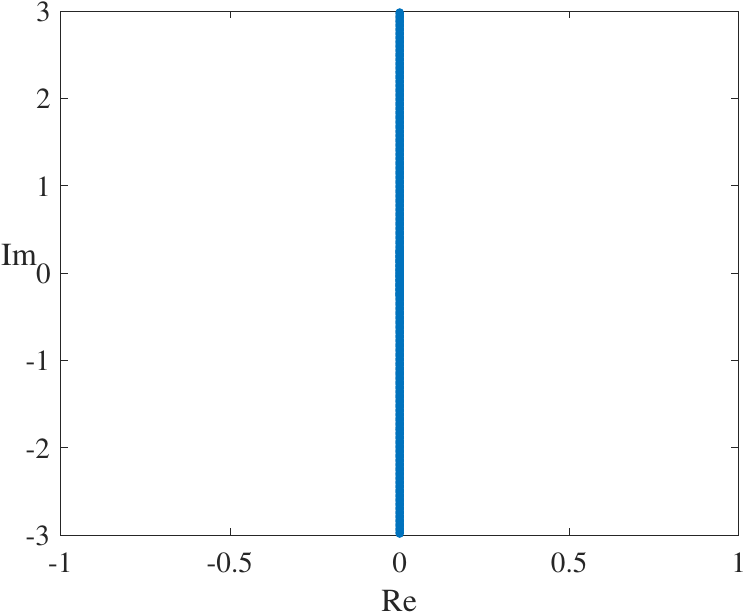}}
	\caption{The same as Figure \ref{fig_13} but for $\alpha_1=\beta_1=1.4$ at point E. }\label{fig_16}
\end{figure}

\begin{figure}[htb!]
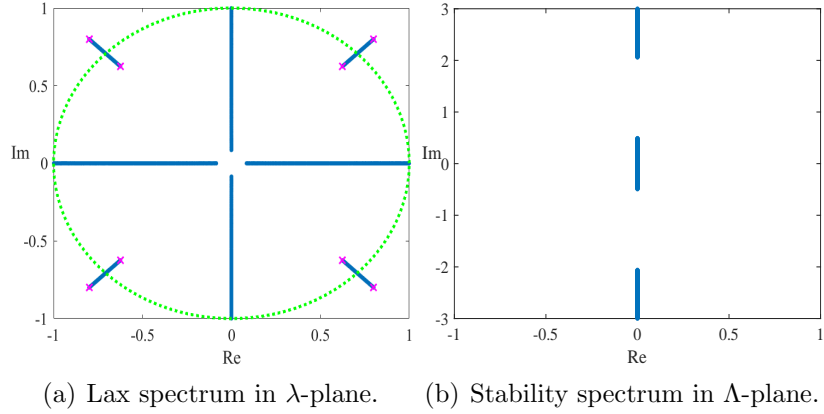

	\centering
	\subfigure[Lax spectrum in $\lambda$-plane.] {\includegraphics[width=2.1in,height=1.9in]{pic_area_1_6_1}}
	\subfigure[Stability spectrum in $\Lambda$-plane.]
	{\includegraphics[width=2.1in,height=1.9in]{pic_area_1_6_2}}
	\caption{The same as Figure \ref{fig_13} but for   $\alpha_1=\beta_1=0.8$ at point F. }\label{fig_17}
\end{figure}

We summarize that every standing periodic wave in region $\rm \Rmnum{1}$ with $\omega \neq 0$ is spectrally unstable due to the figure-eight instability band. On the other hand, every standing periodic wave in region $\rm \Rmnum{1}$ with $\omega =0$ is spectrally stable.

\subsection{Standing periodic waves in region $\rm \Rmnum{2}_A$}

By Proposition \ref{prop-roots}, roots of $P(\lambda)$ form two complex quadruplets located on the unit circle $\mathbb{S}^1$.  Hence we write $\lambda_1=\alpha_1+i\beta_1$ and 
$\lambda_2= \alpha_2 + i\beta_2$ with $\alpha_1^2 + \beta_1^2 = \alpha_2^2 + \beta_2^2 = 1$ so that equations (\ref{521}) yield
\begin{equation*}
\omega =\frac{1}{2} (\alpha_1^2+\alpha_2^2-\beta_1^2-\beta_2^2)
\end{equation*}
and
\begin{equation*}
b = -(\alpha_1^2-\alpha_2^2) (\beta_1^2-\beta_2^2) = (\beta_1^2-\beta_2^2)^2.
\end{equation*}
By using (\ref{u1-u4}), we obtain 
\begin{equation*}
\begin{cases}
u_1=(\beta_1+\beta_2)^2,\\
u_2=(\beta_1-\beta_2)^2,\\
u_3=-(\alpha_1-\alpha_2)^2,\\
u_4=-(\alpha_1+\alpha_2)^2,
\end{cases}
\end{equation*}
which satisfy the same ordering $u_4\leq u_3\leq 0\leq u_2\leq u_1$ as in region \Rmnum{1}. The exact  periodic solution of the first-order invariant (\ref{exact_1}) with this ordering is still written in the same explicit form (\ref{solution_1}). It follows from the phase portraits on Figure \ref{fig_w_16} for $\omega \in (-1,1)$ that the mapping $x \mapsto \theta(x)$ is periodic for the periodic orbits inside the heteroclinic orbits. The values of $\theta(x)$ can be computed from the same formula (\ref{theta}).

\begin{remark}
As $b \to \min\{(1-\omega)^2, (1+\omega)^2\}$, the point inside the existence region $\rm \Rmnum{2}_A$ approaches the boundaries given by the red line for $\omega \in (0,1)$ and by the green line for $\omega \in (-1,0)$ in Figure \ref{fig00}. 
At each boundary, bifurcations of the roots of $P(\lambda)$ occur as follows.
\begin{itemize}
\item If $\omega \in (0,1)$ and $b \to (1-\omega)^2$, then $\beta_1 \rightarrow 0$ which implies that $\lambda_1 \rightarrow 1$. Hence, in this limit one quadruplet is still located on $\mathbb{S}^1$ but the other becomes double real eigenvalues at $\pm 1$. The solution cannot be continued inside region {\rm \Rmnum{4}}. \\

\item If $\omega \in (-1,0)$ and $b \to (1+\omega)^2$, then 
$\alpha_2 \to 0$ which implies that $\lambda_2 \to i$. Hence, in this limit one quadruplet is still located on $\mathbb{S}^1$ but the other becomes double purely imaginary eigenvalues at $\pm i$. The solution continued inside region $\rm \Rmnum{2}_B$ has one complex quadruplet on $\mathbb{S}^1$ and two pairs of 
purely imaginary eigenvalues symmetrically reflected about $\pm i$.
\end{itemize}
\end{remark}

Figure \ref{fig_19} shows numerically computed Lax and stability spectra at point H, for which we take  $\alpha_1=0.6$ and $\alpha_2=0.5$. Besides the Lax spectrum on $\R \cup i\R$, there exists four bands on $\mathbb{S}^1$ in between the two complex quadruplets in the roots of $P(\lambda)$. The stability spectrum includes a segment on the real axis related to the four bands of the Lax spectrum on $\mathbb{S}^1$.

Figure \ref{fig_18} shows similar Lax and stability spectra at point $G$, 
for which $\alpha_1=0.9$ and $\alpha_2=0.2$. Since the point $G$ is close to the boundary $b = (1+\omega)^2$ for $\omega \in (-1,0)$, the bands of the Lax spectrum on $\mathbb{S}^1$ are wider and four roots of $P(\lambda)$ are close to 
the points $\pm i$ . The stability spectrum includes a larger segment along the real axis. 

We summarize that every standing periodic wave in region $\rm \Rmnum{2}_A$ is spectrally unstable due to the instability band on $\R$. 

\begin{figure}[htb!]
	\centering
	\subfigure[Lax spectrum in $\lambda$-plane.] {\includegraphics[width=2.1in,height=1.9in]{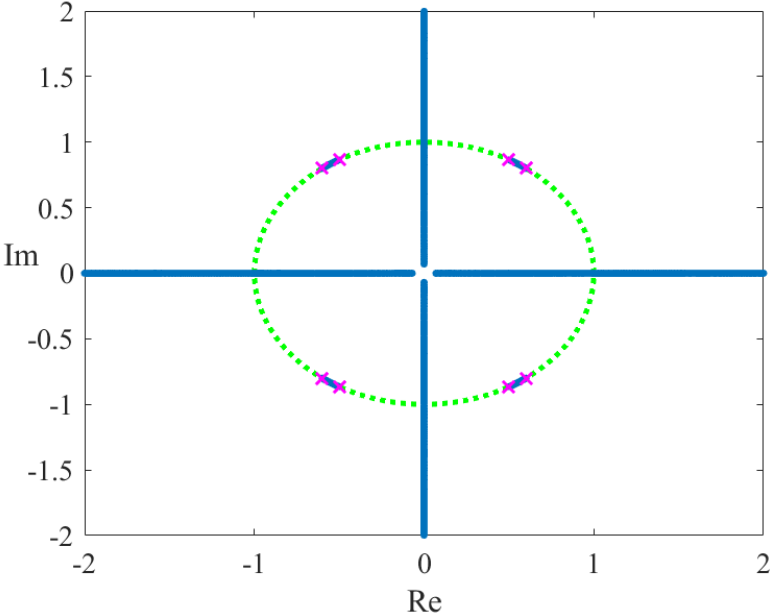}}
	\subfigure[Stability spectrum in $\Lambda$-plane.]
	{\includegraphics[width=2.1in,height=1.9in]{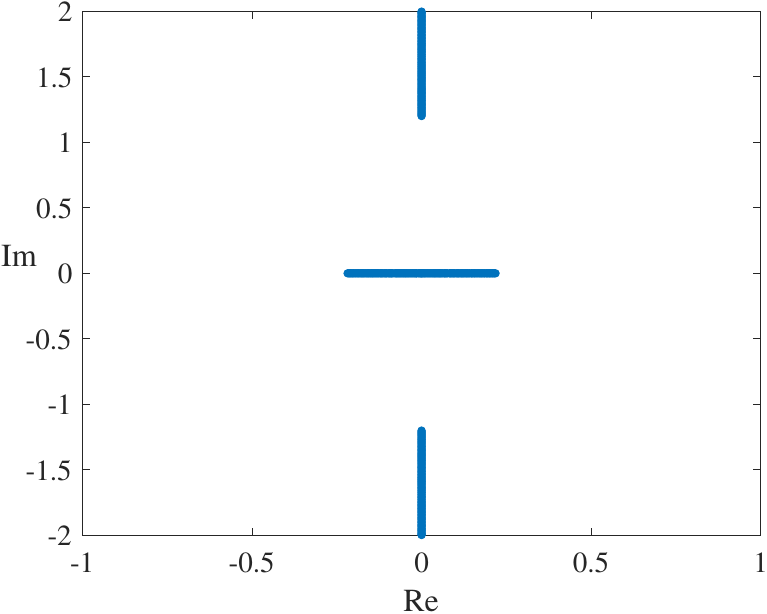}}
	\caption{The same as Figure \ref{fig_13} but for  $\alpha_1=0.6$ and $\alpha_2=0.5$ at point H. }\label{fig_19}
\end{figure}

\begin{figure}[htb!]
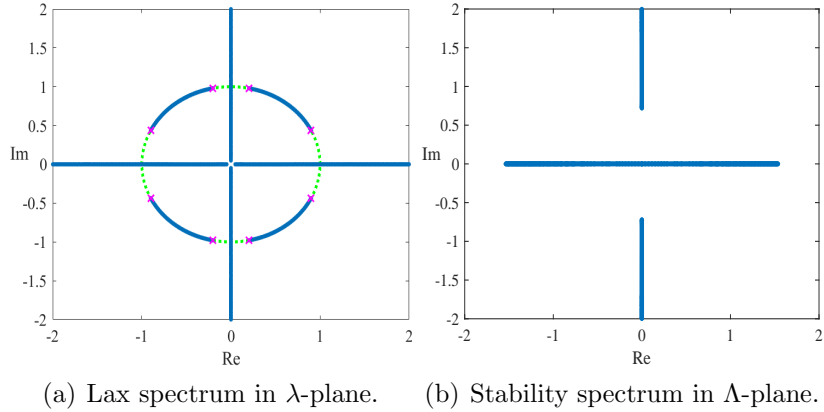

	\centering
	\subfigure[Lax spectrum in $\lambda$-plane.] {\includegraphics[width=2.1in,height=1.9in]{pic_area_2a_1_1}}
	\subfigure[Stability spectrum in $\Lambda$-plane.]
	{\includegraphics[width=2.1in,height=1.9in]{pic_area_2a_1_2}}
	\caption{The same as Figure \ref{fig_13} but for  $\alpha_1=0.9$ and $\alpha_2=0.2$ at point G. }\label{fig_18}
\end{figure}

\subsection{Standing periodic waves in region $\rm \Rmnum{2}_B$}

By Proposition \ref{prop-roots}, roots of $P(\lambda)$ form a complex quadruplet located on the unit circle $\mathbb{S}^1$ and two pairs of purely imaginary eigenvalues symmetrically reflected about $\pm i$.  Hence we write $\lambda_1=\alpha_1+i\beta_1$ and 
$\lambda_2 = i\beta_2$ with $\alpha_1^2 + \beta_1^2 = 1$ so that equations (\ref{521}) yield
\begin{equation*}
\omega = -\frac{1}{4}(\beta_2^2 + \beta_2^{-2} - 2\alpha_1^2+2\beta_1^2)
\end{equation*}
and
\begin{equation*}
b = \frac{1}{16} \left( \beta_2^2 + \beta_2^{-2} + 2-4 \beta_1^2\right)^2.
\end{equation*}
By using (\ref{u1-u4}), we obtain 
\begin{equation*}
\begin{cases}
u_1 = \frac{1}{4}(2\beta_1 + \beta_2 + \beta_2^{-1})^2,\\
u_2 = \frac{1}{4}(2\beta_1 - \beta_2 - \beta_2^{-1})^2,\\
u_3 = -\frac{1}{4}(2\alpha_1-i(\beta_2-\beta_2^{-1}))^2,\\
u_4 = -\frac{1}{4}(2\alpha_1+i(\beta_2-\beta_2^{-1}))^2,
\end{cases}
\end{equation*}
which satisfy the same ordering $0 \leq u_2\leq u_1$ with $u_3$ and $u_4$ being complex conjugate to each other. For definiteness, we will define $u_3 = \gamma + i \eta$ and $u_4 = \gamma - i \eta$ with 
$$
\gamma = -\alpha_1^2 + \frac{1}{4} (\beta_2 - \beta_2^{-1})^2, \quad 
\eta =  \alpha_1 (\beta_2 - \beta_2^{-1}).
$$
The exact  periodic solution of the first-order invariant (\ref{exact_1}) with this ordering of roots of $R(\xi)$ can be written in the explicit form, see \cite{CP_1},
\begin{equation}
\label{solution_2}
{\xi}(x)={u}_1+\frac{({u}_2-{u}_1)\big(1-{\rm cn}(\mu x;k)\big)}{1+\delta+(\delta-1){\rm cn}(\mu x;k)},
\end{equation}
where positive parameters $\delta$, $\mu$ and $k$ are uniquely expressed by  
\begin{align*}
\delta&=\frac{\sqrt{({u}_2-\gamma)^2+\eta^2}}{\sqrt{({u}_1-\gamma)^2+\eta^2}},\\
\mu&=\sqrt[4]{[({u}_1-\gamma)^2+\eta^2][({u}_2-\gamma)^2+\eta^2]},\\
2k^2&=1-\frac{({u}_1-\gamma)({u}_2-\gamma)+\eta^2}{\sqrt{[({u}_1-\gamma)^2+\eta^2][({u}_2-\gamma)^2+\eta^2]}}.
\end{align*}
The periodic solution ${\xi}(x)$ in (\ref{solution_2}) is located in the interval $[u_2,u_1]$ and has period $L=4K(k)\mu^{-1}$. It follows from the phase portraits on Figures \ref{fig_w_15}, \ref{fig_w_1}, and \ref{fig_w_16} for $\omega \in (-\infty,0)$ that the mapping $x \mapsto \theta(x)$ is periodic for the periodic orbits inside the heteroclinic orbits. The values of $\theta(x)$ can be computed from the same formula (\ref{theta}).

\begin{remark}
	As $b\to(1-\omega)^2$ or $b \to (1+\omega)^2$, the point inside the existence region $\rm \Rmnum{2}_B$ approaches the boundary given by the red line in the former limit and by the blue line in the latter limit in Figure \ref{fig00}. At each boundary, bifurcations of the roots of $P(\lambda)$ occur as follows.
	\begin{itemize}
		\item If $b \to (1-\omega)^2$, then $\beta_1 \rightarrow 0$ which implies that $\lambda_1 \rightarrow 1$. Hence, in this limit there still exist two pairs of purely imaginary eigenvalues symmetrically reflected about $\pm i$ but the complex quadruplet becomes double real eigenvalues at $\pm 1$. The solution cannot be continued inside region {\rm \Rmnum{4}}. \\
		
		\item If $b \to (1+\omega)^2$, then 
		$\alpha_1 \to 0$ which implies that $\lambda_1 \to i$. Hence, in this limit there still exist two pairs of purely imaginary eigenvalues symmetrically reflected about $\pm i$ but the complex quadruplet becomes double purely imaginary eigenvalues at $\pm i$. The solution continued inside region $\rm \Rmnum{3}$ has four pairs of purely imaginary eigenvalues symmetrically reflected about $\pm i$.
	\end{itemize}
\end{remark}

%

Figure \ref{fig_21} shows numerically computed Lax and stability spectra at point I, for which we take $\alpha_1=0.95$ and $\beta_2=1.5$. The Lax spectrum includes $\R$, bands on $i\R$ between four roots of $P(\lambda)$ and $\lambda = 0$, and two bands on $\mathbb{S}^1$ between the complex quadruplet of roots of  $P(\lambda)$. The stability spectrum includes a segment on the real axis related to the bands of the Lax spectrum on $\mathbb{S}^1$. 


We summarize that every standing periodic wave in region $\rm \Rmnum{2}_B$ is spectrally unstable due to the instability band on $\R$. 


\begin{figure}[htb!]
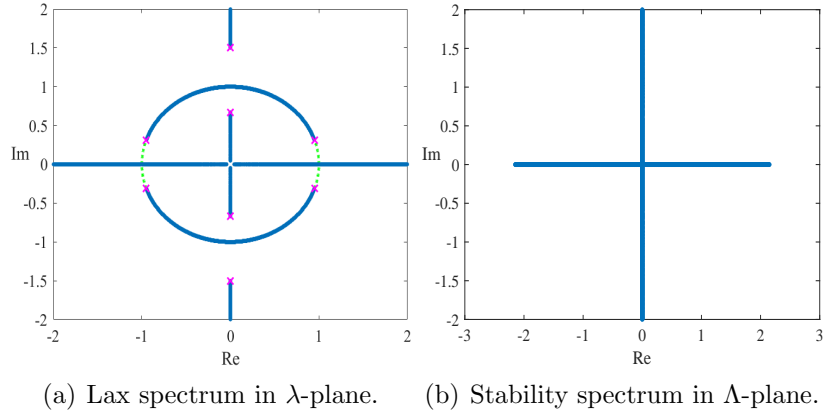

	\centering
	\subfigure[Lax spectrum in $\lambda$-plane.] {\includegraphics[width=2.1in,height=1.9in]{pic_area_2b_2_1}}
	\subfigure[Stability spectrum in $\Lambda$-plane.]
	{\includegraphics[width=2.1in,height=1.9in]{pic_area_2b_2_2}}
	\caption{The same as Figure \ref{fig_13} but for $\alpha_1=0.95$ and $\beta_2=1.5$ at point I. }\label{fig_21}
\end{figure}

\subsection{Standing periodic waves in region \Rmnum{3}}

By Proposition \ref{prop-roots}, roots of $P(\lambda)$ form four pairs of purely imaginary eigenvalues symmetrically reflected about $\pm i$.  Hence we write $\lambda_1 = i\beta_1$ and 
$\lambda_2 = i\beta_2$ with $\beta_2 > \beta_1 > 1$ so that equations (\ref{521}) yield
\begin{equation*}
\omega = -\frac{1}{4} (\beta_1^2 + \beta_1^{-2} + \beta_2^2 + \beta_2^{-2})
\end{equation*}
and
\begin{equation*}
b = \frac{1}{16}(\beta_1^2 + \beta_1^{-2} - \beta_2^2 - \beta_2^{-2})^2.
\end{equation*}
By using (\ref{u1-u4}), we obtain 
\begin{equation*}
\begin{cases}
u_1=\frac{1}{4}(\beta_1- \beta_1^{-1} +\beta_2-\beta_2^{-1})^2,\\
u_2=\frac{1}{4}(\beta_1+ \beta_1^{-1} +\beta_2+\beta_2^{-1})^2,\\
u_3=\frac{1}{4}(\beta_1- \beta_1^{-1} -\beta_2+\beta_2^{-1})^2,\\
u_4=\frac{1}{4}(\beta_1+ \beta_1^{-1} -\beta_2-\beta_2^{-1})^2,
\end{cases}
\end{equation*}
which satisfy the ordering $0\leq{u}_4\leq{u}_3\leq{u}_2\leq{u}_1$. 
Two periodic solutions of the first-order invariant (\ref{exact_1}) exist for this ordering of roots of $R(\xi)$. One solution is given by (\ref{solution_1}) 
with $\xi(x)$ located in the interval $[u_2,u_1]$. Another solution is 
obtained from (\ref{solution_1}) by exchanging $u_1$ with $u_3$ and  $u_2$ with $u_4$, 
\begin{equation}\label{solution_3}
\xi(x)=u_2+\frac{(u_2-u_3)(u_2-u_4)}{(u_2-u_4)+(u_3-u_4){\rm sn}^2(\nu x;k)},
\end{equation}
where $\nu$ and $k$ are exactly the same. The periodic solution $\xi(x)$ in (\ref{solution_3}) is located in the interval $[u_4,u_3]$ and has the same period $L=2K(k)\nu^{-1}$ as (\ref{solution_1}).

It follows from the phase portraits on the left panel of Figure \ref{fig_w_15} that the mapping $x \mapsto \theta(x)$ is monotonically increasing in each period for each of the periodic solutions. The values of $\theta(x)$ can be computed from the same formula (\ref{theta}).

\begin{figure}[htb!]
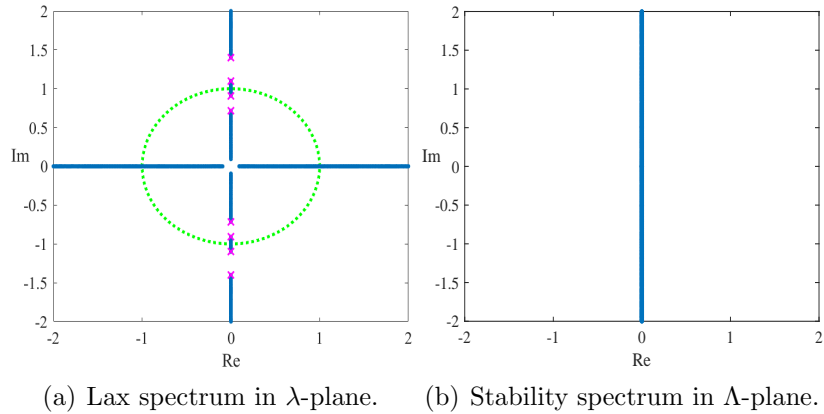

	\centering
	\subfigure[Lax spectrum in $\lambda$-plane.] {\includegraphics[width=2.1in,height=1.9in]{pic_area_3_1_1}}
	\subfigure[Stability spectrum in $\Lambda$-plane.]
	{\includegraphics[width=2.1in,height=1.9in]{pic_area_3_1_2}}
	\caption{The same as Figure \ref{fig_13} but for $\beta_1=1.1$ and $\beta_2=1.4$ at point J. }\label{fig_11}
\end{figure}

Figure \ref{fig_11} shows numerically computed Lax and stability spectra at point J, for which we take $\beta_1=1.1$ and $\beta_2=1.4$. The Lax spectrum includes $\R$ and bands on $i\R$ between eight roots of $P(\lambda)$. The stability spectrum is located on $i \R$. The Lax and stability spectra remain very similar for every point in region $\rm \Rmnum{3}$.

We summarize that every standing periodic wave in region $\rm \Rmnum{3}$ is spectrally stable.

\section{Conclusion}
\label{conclusion}

We have clarified the spectral stability of the standing periodic waves for the MTM system (\ref{ini_1}). The standing waves are written in the form (\ref{standing-wave}) with the wave frequency $\omega \in \R$ and they can be extended as the traveling waves with the wave speed $c \in (-1,1)$ under the Lorentz transformation (\ref{MTM-Lorentz}). The existence of standing 
periodic waves with the reduction $V = \bar{U}$ was obtained in regions \Rmnum{1}, \Rmnum{2}, and \Rmnum{3} of the parameter plane $(b,\omega)$ shown in Figure \ref{fig00}, where $b$ is the constant value of the Hamiltonian function for the spatial Hamiltonian system. 

The Lax spectrum was obtained numerically as the Floquet spectrum of the spectral problem in the Lax pair (\ref{lax_1}). The spectral bands which determine stability or instability of the standing periodic waves 
are located between eight roots of the characteristic function $P(\lambda)$. 
Depending on the point in each region of the parameter plane $(b,\omega)$, 
roots of $P(\lambda)$ appear either as two quadruplets of complex eigenvalues outside the unit circle $\mathbb{S}^1$, or as two quadruplets of complex eigenvalues on $\mathbb{S}^1$, or as one quadruplet on $\mathbb{S}^1$ and two pairs of purely imaginary eigenvalues, or as four pairs of purely imaginary eigenvalues. 

We have found numerically that the standing periodic waves are spectrally stable 
either if the two quadruplets are located at the diagonals of the complex $\lambda$ plane or if all roots of $P(\lambda)$ are purely imaginary. In other cases, 
the standing periodic waves are spectrally unstable either because of the figure-eight instability band or because of the instability band on the real axis. The stability conclusion 
relies on the relation $\Lambda = \pm i \sqrt{P(\lambda)}$ between 
Lax and stability spectra. 

In the particular limit, periodic solutions degenerate to the constant-amplitude solutions and we have found that the constant-amplitude solutions for the saddle points of the spatial Hamiltonian system are spectrally stable. Related to this fact, the solitary wave solutions associated with the heteroclinic orbits of the spatial Hamiltonian system to the saddle points are found to be spectrally stable. However, the constant-amplitude solutions for the center points of the spatial Hamiltonian system are spectrally unstable. 

As further problems, one can extend the analytical and numerical studies to the standing waves in a more general setting with $V \neq \bar{U}$, as well as to try to prove some of the stability and instability conclusions analytically. These problems will wait for further work.

\appendix
\section{Numerical method for calculating the Lax spectrum}
\label{app_1}

Let us consider the first equation of the Lax pair (\ref{eig_1}) rewritten as the spectral problem
\begin{equation}\label{lax_3}
	\begin{aligned}
				\psi_{x}=L(U,V,\lambda)\psi,
	\end{aligned}
\end{equation}
where
$$L=\frac{i}{4}(\lambda^2-\frac{1}{\lambda^2})\sigma_3
-\frac{i\lambda}{2}
\left(\begin{array}{cc}
	0       & \bar{V}  \\
	V & 0 \\
\end{array}\right)
+\frac{i}{2\lambda}\left(
\begin{array}{cc}
	0       &\bar{U}  \\
	U & 0 \\
\end{array}
\right)
+\frac{i}{4}(|U|^2-|V|^2)\sigma_3.
$$
The spectral problem (\ref{lax_3}) is transformed into a linear eigenvalue problem with the help of the following auxillary variables:
	\begin{equation}\label{lax_x_1}
		\psi_0=\lambda \psi, \quad \psi_1=\frac{1}{\lambda}\psi  
		\quad {\rm and} \quad \psi_2=\frac{1}{\lambda}\psi_1.
	\end{equation}
For the extended choice of variables, the spectral problem (\ref{lax_3}) is written as the following linear eigenvalue problem 
	\begin{equation}\label{lax_x_4}
		\left(
		\begin{array}{cccccccc}
			0 &2\bar{V}   & f-4i\partial_x  &0  & 0 &-2\bar{U}  & 1 & 0\\
			-2V &0          &0&f+4i\partial_x   &2U  &  & 0&1\\
			1 & 0 & 0 & 0  & 0 & 0 & 0& 0\\
			0 & 1 & 0 & 0  & 0 & 0 & 0& 0\\
			0 & 0 & 1 & 0  & 0 & 0 & 0& 0\\
			0 & 0 & 0 & 1  & 0 & 0 & 0& 0\\
			0 & 0 & 0 & 0  & 1 & 0 & 0& 0\\
			0 & 0 & 0 & 0  & 0 & 1 & 0& 0\\
		\end{array}
		\right)\left(
		\begin{array}{c}
			\psi_0 \\
			\psi \\
			\psi_1 \\
			\psi_2 \\
		\end{array}
		\right)=\lambda\left(
		\begin{array}{c}
			\psi_0 \\
			\psi \\
			\psi_1 \\
			\psi_2 \\
		\end{array}
		\right),
	\end{equation}
where $f=|V|^2-|U|^2$. Indeed, multiplying both sides of (\ref{lax_3}) by $4\sigma_3$ and using (\ref{lax_x_1}), we have
\begin{align*}
		4\sigma_3\partial_x\psi &= i(\lambda^2-\frac{1}{\lambda^2})\psi+{2i\lambda}
		\left(\begin{array}{cc}
			0       & -\bar{V}  \\
			V & 0 \\
		\end{array}\right)\psi
	+\frac{2i}{\lambda}\left(
	\begin{array}{cc}
	0       &\bar{U}  \\
	-U & 0 \\
\end{array}
\right)\psi-if\psi\\
&= i\lambda\psi_0-i\psi_2+{2i}
\left(\begin{array}{cc}
	0       & -\bar{V}  \\
	V & 0 \\
\end{array}\right)\psi_0
+{2i}\left(
\begin{array}{cc}
	0       &\bar{U}  \\
	-U & 0 \\
\end{array}
\right)\psi_1-if\psi,
\end{align*}
which yields
\begin{equation}\label{lax_x_2}
		\psi_2-4i\sigma_3\partial_x\psi+{2}
		\left(\begin{array}{cc}
			0       & \bar{V}  \\
			-V & 0 \\
		\end{array}\right)\psi_0+2\left(
	\begin{array}{cc}
	0       &-\bar{U}  \\
	U & 0 \\
\end{array}
\right)\psi_1+f\psi
		=\lambda\psi_0.
\end{equation}
Combining (\ref{lax_x_1}) and  (\ref{lax_x_2}) 
yields (\ref{lax_x_4}).

\begin{remark}
	Assume that $U$ and $V$ are periodic in $x$ with the same period $L$ 
	and 
	$$
	f=|V|^2-|U|^2=0.
	$$ 
	According to Floquet's Theorem, the bounded solutions of the linear equation (\ref{lax_x_4}) can be represented in the following form 
$$
\left(
\begin{array}{c}
\psi_0(x) \\
\psi(x) \\
\psi_1(x) \\
\psi_2(x) \\
\end{array}
\right) = 
\left(
\begin{array}{c}
\hat{\psi}_0(x) \\
\hat{\psi}(x) \\
\hat{\psi}_1(x) \\
\hat{\psi}_2(x) \\
\end{array}
\right) e^{i\mu x}, 
$$
where $\hat{\psi}(x)=\hat{\psi}(x+L)$ and $\mu\in\left[-\frac{\pi}{L},\frac{\pi}{L}\right]$.
The linear eigenvalue problem is now rewritten in the following form:
	\begin{equation}\label{lax_x_5}
	\left(
	\begin{array}{cccccccc}
		0 &2\bar{V}   & -4i(\partial_x+i\mu)  &0  & 0 &-2\bar{U}  & 1 & 0\\
		-2V &0          &0&4i(\partial_x+i\mu)   &2U  &  & 0&1\\
		1 & 0 & 0 & 0  & 0 & 0 & 0& 0\\
		0 & 1 & 0 & 0  & 0 & 0 & 0& 0\\
		0 & 0 & 1 & 0  & 0 & 0 & 0& 0\\
		0 & 0 & 0 & 1  & 0 & 0 & 0& 0\\
		0 & 0 & 0 & 0  & 1 & 0 & 0& 0\\
		0 & 0 & 0 & 0  & 0 & 1 & 0& 0\\
	\end{array}
	\right)\left(
	\begin{array}{c}
		\hat{\psi}_0 \\
		\hat{\psi}\\
		\hat{\psi}_1\\
		\hat{\psi}_2\\
	\end{array}
	\right)=\lambda\left(
	\begin{array}{c}
\hat{\psi}_0 \\
\hat{\psi}\\
\hat{\psi}_1\\
\hat{\psi}_2\\
	\end{array}
	\right).
\end{equation}
The Fourier collocation method, see \cite[Chapter 2,p.45]{Yjk}, can be used to solve the linear eigenvalue problem  (\ref{lax_x_5}) with the Floquet parameter $\mu$. Tracing the set of eigenvalues $\lambda$ for $\mu\in\left[-\frac{\pi}{L},\frac{\pi}{L}\right]$ gives the band of the Floquet spectrum in the $\lambda$ plane. 
\end{remark}

\begin{remark}
The Fourier collocation method from \cite{Yjk} is similar to Hill's method used in \cite{CD_1,DK_1}. Both of them use Fourier series to approximate $U$, $V$ and the eigenfunctions. 
\end{remark}

\begin{remark}
For given functions $U$ and $V$ defined on $\R$, the Chebyshev collocation method \cite{CW_1} can be used to solve the linear eigenvalue problem  (\ref{lax_x_4}). However, this method seems to be not applicable to periodic domain.	On the other hand, the finite difference method can also be used to calculate the linear eigenvalue problem (\ref{lax_x_4}), see \cite{CPW_1} for a similar study. However, the accuracy of the finite difference method is poor.
\end{remark}

\section*{Declaration of competing interest}
The authors declare that they have no competing financial interests or personal relationships that could have influenced the work reported in this paper.

\section*{Data availability}
The data for this study are available from the authors upon  request.

 \end{document}